%% file: main-full-1.tex
		\newcommand\FINAL[1]{} 

\providecommand\apxparameter{appendix=inline}

\providecommand\FINAL[1]{#1}
\newcommand\PASFINAL[1]{#1}
\FINAL{\renewcommand\PASFINAL[1]{}}

\newcommand\todoinlineD[1]{\FINAL{}}

\newcommand\ICALPSHORTER[1]{}
\newcommand\ICALPDETAILED[1]{}
\newcommand\ICALPSHORTERJOHANNE[1]{}

 \newcommand\NNCOUNT[1]{\NN_{\mathrm{Count}(#1)}}

\renewcommand\P{\operatorname{P}}
\newcommand\NP{\operatorname{NP}}
\newcommand\coNP{\operatorname{coNP}}
\newcommand\PP{\operatorname{PP}}

\newcommand\SAT{\operatorname{SAT}}

\newcommand\RP{\mathbb{R}_{\ge 0}}

\documentclass[11pt,a4paper]{article}
\pdfoutput=1

\usepackage{versions}
\usepackage{adjustbox}
 
\excludeversion{Skipconnections}

\newcommand\Netnotation[4]{\mathcal{N} \mathcal{N}^{#1}_{#2,#3,#4}}

\newcommand\lead{\textrm{lead}}
\newcommand\shift{\textrm{shift}}

\newcommand\NNlead{\NN_{\lead}}
\newcommand\NNshift{\NN_{\shift}}

\newcommand\test{\textrm{test}}

 \newcommand\NNtest{\NN_{\mathrm{test}}}

\newcommand\ellit{\mathcal{E}}

\newcommand\NNT[1]{\NN_{T(#1)}}

\newcommand{\myrect}[5]{%
    \draw (#1,#2) rectangle (#1+#3,#2+#4);
    \node at (#1 + #3/2, #2 + #4/2) {\small #5};
}

\newcommand{\myrectrotate}[5]{%
    \draw (#1,#2) rectangle (#1+#3,#2+#4);
    \node[rotate=90] at (#1 + #3/2, #2 + #4/2) {\small #5};
}

\usepackage[T1]{fontenc}
\usepackage[utf8]{inputenc}
\usepackage[UKenglish]{babel}
\usepackage{lmodern}
\usepackage{microtype}
\usepackage[margin=2.5cm]{geometry}
\usepackage{amsmath,amssymb,amsthm}
\usepackage{graphicx}
\usepackage{xcolor}

\usepackage{url}
\PassOptionsToPackage{dvipsnames}{xcolor}

\usepackage[\apxparameter,bibliography=common]{apxproof}

\theoremstyle{plain}
\newtheoremrep{theorem}{Theorem}[section]
\newtheoremrep{lemma}[theorem]{Lemma}
\newtheoremrep{proposition}[theorem]{Proposition}
\newtheoremrep{corollary}[theorem]{Corollary}
\newtheoremrep{claim}[theorem]{Claim}
\newtheoremrep{mytheorem}[theorem]{Theorem}

\theoremstyle{definition}
\newtheorem{definition}[theorem]{Definition}
\newtheorem{remark}[theorem]{Remark}
\theoremstyle{plain}

\ifdefined\proofsketch\else
  \newenvironment{proofsketch}{\begin{proof}[Proof sketch]}{\end{proof}}
\fi

\usepackage{amscd}
\usepackage{url}
\usepackage{multicol}

\usepackage{amsmath}
\usepackage{amsthm}
\usepackage{graphicx}
\usepackage{algpseudocode}
\FINAL{\usepackage[colorinlistoftodos]{todonotes}}
\PASFINAL{\newcommand\todo[1]{}}

\newcommand\R{\mathbb{R}}

\newcommand\onethreefouradic{{\mathbb{\mathbb{DD}}_{1,3}}}
\newcommand\onethreefouradicn[1]{\onethreefouradic^{(#1)}}
\newcommand\DDp{\onethreefouradic}

\usepackage{tikz,xcolor,hyperref}
\usetikzlibrary{arrows.meta}
\usepackage{pgfplots}
\pgfplotsset{compat=1.18}

\usetikzlibrary{arrows.meta}

\definecolor{lime}{HTML}{A6CE39}
\DeclareRobustCommand{\orcidicon}{
	\begin{tikzpicture}
	\draw[lime, fill=lime] (0,0) 
	circle [radius=0.16] 
	node[white] {{\fontfamily{qag}\selectfont \tiny ID}};
	\draw[white, fill=white] (-0.0625,0.095) 
	circle [radius=0.007];
	\end{tikzpicture}
	\hspace{-2mm}
}

\foreach \x in {A, ..., Z}{%
	\expandafter\xdef\csname orcid\x\endcsname{\noexpand\href{https://orcid.org/\csname orcidauthor\x\endcsname}{\noexpand\orcidicon}}
}

\title{Fractal Gadgets for Neural Networks: The Complexity of the Narrow Regime}

\author{
Olivier Bournez\thanks{Institut Polytechnique de Paris, Ecole Polytechnique, LIX, 91128 Palaiseau Cedex, France. \texttt{olivier.bournez@lix.polytechnique.fr}}
\and Johanne Cohen\thanks{LISN, CNRS, Universit\'e Paris-Saclay, 91190 Gif-sur-Yvette, France. \texttt{johanne.cohen@lisn.fr}}
\and Laura Cohen\thanks{CentraleSupelec, Universit\'e Paris-Saclay, 91190 Gif-sur-Yvette, France. \texttt{laura.cohen@student-cs.fr}}
\and Adrian Wurm\thanks{Computer Science Institute, BTU Cottbus-Senftenberg, Platz der Deutschen Einheit 1, D-03046 Cottbus, Germany. \texttt{wurm@b-tu.de}}
}

\newcommand\NN{\mathcal{N}}

\newcommand\ReLU{\operatorname{ReLU}}

\begin{document}

\maketitle

\begin{abstract}
We study the verification problem for deep narrow ReLU neural networks:
given a network $\NN\colon[0,1]\to\mathbb{R}$ of bounded width computing
an iterated piecewise-affine map, does there exist an input
$x\in[0,1]$ satisfying a prescribed output constraint? Classical
hardness results for ReLU verification rely on reductions from
\textsc{3-Sat} with one neuron per Boolean variable, and therefore
do not address the narrow regime in which width is fixed to a small
constant. At the other end of the scale, a width-$1$ ReLU network
computes a monotone piecewise-affine function with polynomially many
pieces, so its verification is solvable in polynomial time; this
remains true for iterated width-$1$ networks under a unary time
bound. The precise threshold at which verification becomes
intractable had remained open.

Our main result is that \emph{verification of ReLU neural networks is
$\NP$-complete at width~$4$ for arbitrary inputs in $[0,1]$}. When
inputs are restricted to a natural discrete encoding set,
$\NP$-completeness already holds at width~$3$,  a baseline building on earlier techniques. Combined with polynomial-time decidability at
width~$1$, this pins down the complexity on the encoding set up to the
single open case of width~$2$; for arbitrary inputs in~$[0,1]$, the
classification is tight up to widths~$2$ and~$3$.

The technical core of the paper is a \emph{fractal preprocessing
gadget} of independent interest: a width-$2$ ReLU subnetwork that,
applied to a real input $x\in[0,1]$, vanishes precisely when $x$ lies
near a finite Cantor-like subset of $[0,1]$. The gadget effectively
reduces verification of a continuous function on $[0,1]$ to
verification on a discrete set of $2^n$ points, without any ad hoc
real encoding of discrete data. This ``discretisation trick'' is the
missing ingredient for width-bounded hardness reductions: combined
with width-$2$ digit-extraction and clause-evaluation subnetworks
sharing a single two-register state, it yields, together with one register for the fractal error, a width-$4$ polynomial-time reduction from \textsc{3-Sat} on all of~$[0,1]$.

The same construction transfers a spectrum of further hardness
results to the width-$3$ setting on the encoding set: the
universal-zero problem is $\coNP$-complete, counting zeros is
$\#\P$-complete, a majority variant is $\PP$-complete, and
approximating the minimum output within any constant gap inherited
from \textsc{Max-3Sat} is $\NP$-hard. The $\NP$-, $\coNP$-, and
inapproximability results lift to arbitrary real inputs in~$[0,1]$ at
width~$4$; the lifting of counting and majority, and whether the
lifting to $[0,1]$ can be achieved at width~$3$, remain open. We
relate the latter to a continuity obstruction for width-$3$ $\ReLU$
gadgets discriminating between valid and invalid digit-$4$ prefixes.
\medskip
\noindent\textbf{Keywords:}Analysis and Verification of AI, Models of computation, Complexity theory, Formal neural networks, Neural network verification

\end{abstract}

%
%

\section{Introduction}

The computational capabilities of neural networks---their expressive
power, their links to classical models of computation, and the effect
of architectural constraints such as depth, width, and activation
function---have long motivated the study of  \emph{verification problems}, which ask whether a
given network satisfies a prescribed input--output property. Such
questions are especially important whenever neural networks are
deployed in safety-sensitive settings. Concretely, the decision problem
we study is the \emph{existential} decision problem: given a network and a linear
output specification (an affine constraint on the output, such as $\NN(x)=0$ or $\NN(x)\ge\tau$), does there exist an input meeting it?  This corresponds to the
\emph{reachability} question at the core of neural-network
verification~\cite{salzer2021reachability}, and is  the basic primitive on which
richer property checks are built. We also consider its universal
counterpart, asking whether \emph{every} input satisfies the
specification, which turns out to be $\coNP$-complete in the same regime.

How hard verification is depends sharply on the architecture. The \emph{width} of a network is the number of neurons per hidden layer, and its \emph{depth} the number of hidden layers (see Section~\ref{sec:def}). Most existing
hardness results address networks of arbitrary width and bounded depth, in
the spirit of classical universality theorems on compact
domains~\cite{cybenko1989approximation,hornik1991approximation}. In this
paper we study the complementary regime of \emph{deep, narrow} networks:
bounded width and unbounded depth. Our goal is to pin down the minimum width at
which verification becomes intractable for the standard activation
function, $\ReLU$. Our results are not intended to provide practical verification
methods for large-scale models; rather, we show that extremely narrow
$\ReLU$ networks already capture the full hardness of program-like
verification, and we identify the precise width at which this transition
occurs.

\paragraph*{Deep, narrow networks.}
Although depth has long been conceptually natural, it became
practically effective only with modern
architectures~\cite{he2016deep}, which popularised \emph{skip
connections} and led to the interpretation of very deep networks as
discrete-time dynamical systems arising from numerical
discretisations of ODEs (PolyNet, FractalNet, ResNet variants,
Neural
ODEs~\cite{chen2018neural,he2016deep,kidger2022neural,larsson2016fractalnet,lu2018beyond}).
These models lie in the deep, narrow regime.

\paragraph*{Why the narrow regime.}
Isolating width is natural for two reasons. A network of bounded width
and unbounded depth is an iterated map on a constant-dimensional
state: a discrete dynamical system in which width plays the role of a
space-like resource that can be studied separately from depth.
Constant-width computation has a distinguished history in this
respect: width-$5$ branching programs already capture
$\mathrm{NC}^1$~\cite{barrington1989bounded}, and reachability for
low-dimensional piecewise-affine maps is a classical source of
hardness and of stubborn open problems~\cite{KCG94}. Our results
locate $\ReLU$ verification within this landscape: the network is the
program, its width is the number of real registers kept alive, and the
question is how few registers already make verification intractable.
This regime is also a reasonable abstraction of very deep architectures such as ResNet variants, FractalNet, and Neural ODEs, which iterate blocks of fixed width within a stage.  Thus,  the width at which verification becomes hard has a concrete
architectural reading rather than being an artefact of an exotic model.
Pinning down that threshold is thus at once a clean complexity question,
a width hierarchy for $\ReLU$ verification, and a statement about
realistic networks.

\paragraph*{Width as a complexity parameter.}
In the absence of width restrictions, verification and reachability for $\ReLU$
networks are $\NP$-complete via reductions from $\SAT$ that use one
neuron per Boolean
variable~\cite{salzer2021reachability,wurm1}; hardness becomes even
stronger ($\exists\R$-hardness) for semi-algebraic or analytic
activations~\cite{wurm1,wurm2}. All these
reductions use width freely, and therefore say nothing about the narrow
regime. The complementary question, what happens at small constant
width, has received less attention. On the expressiveness side,
universal approximation holds already at constant widths barely above
the input
dimension~\cite{hanin2017approximating,lu2017expressive,park2021minimum,hong2024bridging,kidger2020universal,zhang2024deep},
and the classical depth separations are themselves witnessed by
width-$2$ sawtooth networks~\cite{telgarsky2016benefits}; these works
address approximation, not verification. Our results show that the
hardness of verification sets in essentially at the width threshold
where expressiveness begins: restricting the width of the analysed
network, a natural structural hope for verification
tools~\cite{katz2017reluplex,salzer2021reachability}, cannot by itself
yield tractability.

For $\ReLU$ networks, the two extremes of the width scale are easy.
At width~$1$, verification is decidable in polynomial time
(Proposition~\ref{prop:width-one}): a width-$1$ $\ReLU$ network
computes a monotone one-dimensional piecewise-affine function with
polynomially many pieces; the iterated variant with a unary time
bound is also tractable (Theorem~\ref{thdimone}). At the other end, \cite[Theorem~23]{MFCS2025} states that the
transition map of a Turing machine, on configurations encoded over a
dense discrete subset, can be implemented at width~$3$.

This leaves a narrow but genuine gap, which concerns the inputs. On
the encoding set, the techniques of~\cite{MFCS2025} point to
hardness at width~$3$, which we establish as a baseline
(Theorem~\ref{thm:np-complete-encoded-zero} and
Remark~\ref{rem:baseline}). But verification quantifies over \emph{all} real inputs, whereas the
construction of~\cite{MFCS2025} works only on exact encodings, and no
continuous map sends $[0,1]$ onto the encoding
set~\cite[Remark~9]{MFCS2025}. Classical $\SAT$-style reductions for
neural networks break down in this regime, since they rely on encoding one
Boolean variable per neuron---a luxury that is no longer available when
width is held constant. The deeper difficulty is not combinatorial but
geometric: the input to a width-$3$ network is a single real
$x\in[0,1]$, whereas a satisfying assignment for a $\textsc{3-Sat}$
instance with $n$ Boolean variables and $m$ clauses is an element of
the discrete set $\{0,1\}^n$. There is no obvious mechanism for
converting a continuous verification problem into a discrete one
without spending
extra width on the conversion itself.

\paragraph*{Main contribution.}
We narrow this gap. Our main result (Theorem~\ref{thm:lifted} and
Corollary~\ref{cor:np-conp-lifted}) is that \emph{verification of
$\ReLU$ neural networks on arbitrary inputs in~$[0,1]$ is
$\NP$-complete at width~$4$}. On the encoding set, the baseline holds
at width~$3$ (Theorem~\ref{thm:np-complete-encoded-zero}). Combined
with the polynomial-time algorithm at width~$1$, this leaves open only
width~$2$ on the encoding set, and widths~$2$ and~$3$ on~$[0,1]$.

The technical core of our result is a \emph{fractal preprocessing gadget}
of independent interest: a continuous piecewise-affine self-map
$F_\varepsilon\colon[0,1]\to[0,1]$, realisable by a width-$2$ $\ReLU$
network, whose $n$-fold iterate vanishes exactly on small neighbourhoods of
a finite Cantor-like subset of $[0,1]$ of size $2^n$. The gadget
effectively reduces verification of a continuous function on $[0,1]$ to
verification on a discrete set of $2^n$ points, without increasing the
width. This ``discretisation'' step is, we argue, the missing ingredient for
width-bounded hardness reductions: once the gadget  is available, $\textsc{3-Sat}$
can be encoded by width-$2$ digit-extraction and clause-evaluation
subnetworks that share a single two-register state
$(x,p)\in[0,1]^2$ throughout the computation; with one more register for the fractal error, this gives a width-$4$ polynomial-time reduction. By contrast, naive SAT-style reductions require
width that grows linearly with the number of Boolean variables.

We stress where the difficulty lies. It is not that a discontinuous
membership test must be approximated by a continuous piecewise-affine
function: this is always possible, at any width. The constraints are
quantitative and algorithmic. The gadget must vanish \emph{exactly} on
the $2^n$ encoding points; away from them it must keep a gap that
survives the clause-evaluation layers applied afterwards, whose
composed Lipschitz constant is exponential, of order $4^{\Theta(mn)}$;
and the gadget, together with its calibrated slope parameter
$\varepsilon$, must be produced in polynomial time at width $2$.
Beating this exponential error amplification with polynomial-size
weights is what the self-similar structure buys
(Section~\ref{sec:fractal}).

\paragraph*{A spectrum of width-$3$ hardness results.}
Once the fractal gadget is available, a broad range of further complexity
results follows for width-$3$ $\ReLU$ verification. On inputs restricted
to the discrete Cantor-like encoding set, the universal-zero problem is
$\coNP$-complete, counting zeros is $\#\P$-complete, deciding whether more
than half of the encodings are zeros is $\PP$-complete, and approximating
the minimum output within any constant additive gap inherited from
$\textsc{Max-3Sat}$ is $\NP$-hard~\cite{hastad2001some}. For each of
the single-quantifier statements ($\NP$, $\coNP$, and inapproximability),
the fractal gadget transports hardness from the discrete input set to
arbitrary real inputs in~$[0,1]$, at the cost of one additional unit of
width (width~$4$): this extra register carries the fractal error
through the clause-evaluation stages, so that it can be combined with
the clause count at the end. The corresponding lifting of counting and
majority hardness is an open problem, as is the question of whether
the lifting to $[0,1]$ can be carried out at width $3$.

\paragraph*{Technical highlights.}
Boolean assignments $\alpha\in\{0,1\}^n$ are encoded as reals
$x(\alpha)=\sum_{i=1}^{n} a_i\,4^{-i}$ with digits $a_i\in\{1,3\}$.
The digit set $\{1,3\}$ provides a uniform separation between
admissible leading-digit intervals in~$[0,1]$, which is what
allows continuous width-$2$ subnetworks to extract, shift, and test
base-$4$ digits robustly. The associated discrete set
$\{x(\alpha)\}_\alpha$ is precisely the set around which
$F_\varepsilon^{[n]}$ vanishes (up to an $\varepsilon$-neighbourhood);
a width-$3$ truncation gadget (Lemma~\ref{lem:trunck}) then
restores well-formed digits from near-misses, so that every real input
producing zero error is effectively projected onto a valid encoding.
With an additional width-$2$ conditional-flag-clearing gadget on top, each
clause can be evaluated within the width budget, yielding a
width-$3$ reduction of depth polynomial in $n+m$.

\paragraph*{Organisation.}
Section~\ref{sec:def} fixes the neural network model and the
$\{1,3\}$-quadradic encoding. Section~\ref{sec:auxiliary} develops
the width-$2$ digit gadgets. Section~\ref{sec:restricted}
establishes $\NP$-completeness on the discrete encoding set and
derives the companion $\coNP$, $\#\P$, $\PP$, and inapproximability
results. Section~\ref{sec:filtering} develops the fractal
preprocessing gadget and lifts the single-quantifier hardness
results to inputs in~$[0,1]$.

\section{Definition}
\label{sec:def}

\subsection{Neural network model}
\noindent\textbf{Network representation.}
Throughout the paper, a \emph{rational affine map}
$\mathbb{R}^n\to\mathbb{R}^m$ is a map $x\mapsto A x+b$ with
$A\in\mathbb{Q}^{m\times n}$ and $b\in\mathbb{Q}^m$; we encode it by the
list of its rational entries, each given as a pair of integers (numerator
and denominator) written in binary. A \emph{$\ReLU$ neuron} is a map
$x\mapsto\ReLU(\ell(x))$, where $\ell\colon\mathbb{R}^n\to\mathbb{R}$ is
a rational affine map and $\ReLU(t)=\max(t,0)$.

A \emph{feedforward $\ReLU$ network} of width $k$ and depth $L$ is a
composition
\(
  A_{\mathrm{out}}\;\circ\;R_L\;\circ\;\cdots\;\circ\;R_1\;\circ\;A_{\mathrm{in}},
\)
where $A_{\mathrm{in}}$ and $A_{\mathrm{out}}$ are rational affine maps
and each hidden layer $R_\ell\colon\mathbb{R}^k\to\mathbb{R}^k$ is a
vector of $k$ $\ReLU$ neurons sharing a common input. The \emph{encoding
size} of such a network is the total bit-length of the rational entries
of $A_{\mathrm{in}}$, $A_{\mathrm{out}}$, and the neurons of
$R_1,\dots,R_L$; throughout, networks are assumed to have encoding size
polynomial in the other input parameters. The following notation,
adapted from~\cite{kidger2020universal}, fixes input, output, and width
simultaneously.

\begin{definition}[$\Netnotation{\ReLU}{n}{m}{k}$]
For $n,m,k\in\mathbb{N}$, we denote by $\Netnotation{\ReLU}{n}{m}{k}$ the
class of functions $\mathbb{R}^n\to\mathbb{R}^m$ computed by feedforward
$\ReLU$ networks with $n$ input neurons, $m$ output neurons, and any
number of hidden layers, each of width exactly $k$.
\end{definition}

We treat unused outputs implicitly: a network in
$\Netnotation{\ReLU}{n}{m}{k}$ computes a function
$\mathbb{R}^n\to\mathbb{R}^{m'}$ for any $m'\le m$ by composing with the
coordinate projection $\mathbb{R}^m\to\mathbb{R}^{m'}$. Composition
behaves well with respect to width: if
$\NN_1\in\Netnotation{\ReLU}{n}{m}{w_1}$ and
$\NN_2\in\Netnotation{\ReLU}{m}{k}{w_2}$, then
\(
  \NN_2\circ\NN_1\;\in\;\Netnotation{\ReLU}{n}{k}{\max(w_1,w_2)},
\)
and we write $\NN^{[t]}$ for the $t$-fold iterated composition of a
network $\NN\in\Netnotation{\ReLU}{n}{n}{k}$ with itself.

\subsection{Encoding discrete information into a real input: the digit-gap principle}
\label{sec:13quadradics}

A recurring challenge in width-bounded hardness reductions is to
\emph{encode discrete information into a single real input} and then
recover it using a $\ReLU$ network of fixed width. The reductions of this
paper all take this shape: a discrete instance (a Boolean assignment,
a clause index, a digit sequence) is mapped to a real $x\in[0,1]$, and
a width-$3$ $\ReLU$ network decodes $x$ digit by digit in some fixed
positional system.

Which positional systems are compatible with this scheme? Any digit extractor realised by a $\ReLU$ network is, by construction, a continuous function of the real
$x\in[0,1]$. This immediately rules out the naive binary encoding
$x=\sum_{i\ge 1}b_i\,2^{-i}$ with $b_i\in\{0,1\}$: the admissible
values of the leading bit partition $[0,1]$ into $[0,\tfrac12)$ and $[\tfrac12,1]$, and any
continuous map that returns $0$ on the first interval and $1$ on the
second would have to be discontinuous at $x=\tfrac12$. The obstruction
is general: no encoding in which admissible values of the leading
digit cover $[0,1]$ up to a measure-zero set admits a continuous digit
extractor at all, let alone one realisable by a fixed-width $\ReLU$ network.

The way around this obstruction, following~\cite{MFCS2025}, is to
force a \emph{gap} between the admissible leading-digit intervals. We
work in base $4$ with digits restricted to $\{1,3\}$: the admissible
leading-digit intervals are then $[\tfrac14,\tfrac12)$ and
$[\tfrac34,1]$, separated by gaps of width $\tfrac14$. Inside a gap,
any continuous function may take constant values, so width-$2$
$\ReLU$ subnetworks can robustly extract, shift, and conditionally
test base-$4$ digits on this encoding (Lemmas~\ref{lem:lead},
\ref{lem:shift}, \ref{lem:test}). Moreover, among the symmetric pairs
$\{a,b\}\subseteq\{0,1,2,3\}$ with $|a-b|\ge 2$ that would provide the
required gap, $\{1,3\}$ is the only choice that leaves the digit $0$ free
to serve as a canonical ``absent digit'' marker---a role played by the
truncation gadget of Lemma~\ref{lem:trunck}, which restores any
leading $0$ with a $\{1,3\}$-digit.

\begin{definition}[{$\{1,3\}$-quadradics} \cite{MFCS2025}] \label{def:quadratic}
We denote by $\onethreefouradic$ the set of dyadic numbers whose
base-$4$ expansion consists of finitely many $1$'s and $3$'s followed
by an infinite sequence of zeros, 
and by $\overline{\onethreefouradic}$
the set of real numbers whose base-$4$ expansion is infinite and uses
only the digits $1$ and $3$:
\(
  \onethreefouradic
  \;=\;\Bigl\{\,d=\sum_{i=-m}^{p}a_i\,4^{-i}\;:\;
  a_i\in\{1,3\},\ m,p\in\mathbb{N}\,\Bigr\}\)
  and \(
  \overline{\onethreefouradic}
  \;=\;\Bigl\{\,d=\sum_{i=-m}^{\infty}a_i\,4^{-i}\;:\;
  a_i\in\{1,3\},\ m\in\mathbb{N}\,\Bigr\}.
\)
Elements of
$\onethreefouradic$ are called
\emph{$\{1,3\}$-quadradics}. For $n\ge 1$, we write
$\onethreefouradicn{n}$ for the set of length-$n$ $\{1,3\}$-quadradics
in $[0,1]$,
\(
  \onethreefouradicn{n}
  \;=\;\Bigl\{\,\sum_{i=1}^{n} a_i\,4^{-i}\;:\;a_i\in\{1,3\}\,\Bigr\},
\)
and set $\onethreefouradicn{0}:=\{0\}$.
\end{definition}

A Boolean assignment $\alpha\in\{0,1\}^n$ will be encoded as
\[
  x(\alpha)\;:=\;\sum_{i=1}^{n} a_i\,4^{-i}\;\in\;\onethreefouradicn{n},
  \qquad
  a_i=\begin{cases}
    1 & \text{if } \alpha(y_i)=1,\\
    3 & \text{if } \alpha(y_i)=0.
  \end{cases}
\]
The map $\alpha\mapsto x(\alpha)$ is a bijection from $\{0,1\}^n$ onto
$\onethreefouradicn{n}$, the set of cardinality $2^n$, and its image is
uniformly separated: any two distinct elements of
$\onethreefouradicn{n}$ differ by at least $2\cdot 4^{-n}$. The same
encoding is also used to carry auxiliary information---clause-satisfaction
flags, accumulated errors---in a second register $p\in[0,1]$
throughout the reduction.

This encoding, however, comes with an inherent geometric complication
that shapes the whole paper. The set
$\overline{\onethreefouradic}\cap[0,1]$ is a compact, totally
disconnected subset of $[0,1]$ with empty interior, in bijection with
$\{1,3\}^{\mathbb{N}}$. Its length-$n$ truncations
$\onethreefouradicn{n}$ are the $n$-th level vertices of a self-similar
Cantor-like construction. A satisfying assignment for
$\textsc{3-Sat}$ thus corresponds to a point in a fractal-like set of
size $2^n$ inside $[0,1]$, while the verification network receives as
input an arbitrary real $x\in[0,1]$, not necessarily an   element of
$\onethreefouradicn{n}$. The central technical problem of the paper is
therefore not the evaluation of clauses---which, once a valid encoding
is in hand, is straightforward---but the preprocessing task of
\emph{certifying} that $x$ lies, or lies near, $\onethreefouradicn{n}$,
using only width-$3$ $\ReLU$ dynamics. Section~\ref{sec:fractal}
solves this problem via a continuous piecewise-affine map whose
$n$-fold iterate has a zero set approximating $\onethreefouradicn{n}$
and which is itself realisable at width-$2$.

\section{Auxiliary Results}
\label{sec:auxiliary} 

The statements in this section are all grounded in a single
compilation principle: any function given as a nested min/max
expression over rational affine forms, which includes every gadget of
this paper, can be realised as a $\ReLU$ network of minimal width,
efficiently and effectively (Lemma~\ref{lem:CPAM}). Neither
ingredient is new in isolation: lattice (min/max) representations of
piecewise-linear functions are
classical~\cite{tarela1999region,ovchinnikov2002max,hanin2017approximating}, and their
realisation by $\ReLU$ networks of unrestricted width is
standard~\cite{arora2016understanding}. The point of
Lemma~\ref{lem:CPAM} is the width discipline: a nested expression
compiles into width~$2$, one layer per level, while keeping the input
register alive; we claim no such compilation for arbitrary continuous
piecewise-affine maps processed in their natural piece order. Every gadget
below is then obtained by \emph{programming with $\ReLU$ networks}:
each one carries out a small arithmetic or logical task on one or two real
registers, and is realised as a feedforward $\ReLU$ network by
combining Lemma~\ref{lem:CPAM} with standard width-preserving
composition.

Concretely, width is the number of real registers a subnetwork keeps
alive at once and depth is the number of operations it performs, with
all control flow folded into the affine coefficients rather than into
branching. Every gadget of this section respects the same budget: its
state fits in at most two registers, with one extra neuron per layer
for working arithmetic, so each lies in
$\Netnotation{\ReLU}{n}{m}{k}$ for $n,m\le 2$ and width $k\le 3$. The
full reduction of Section~\ref{sec:restricted} is then obtained by
sequentially composing these gadgets on a shared two-register state.

\begin{lemmarep}[Compilation of nested min/max expressions]  \label{lem:CPAM}
Let $f\colon\RP\to\mathbb R$ be given by a nested expression
$f=\tau_1(\ell_1,\tau_2(\ell_2,\dots,\tau_{k-1}(\ell_{k-1},\ell_k)\dots))$,
where each $\ell_i$ is a rational affine form in $x$ and each
$\tau_i\in\{\min,\max\}$. Then $f$ can be computed by a network in $\Netnotation{\ReLU}{1}{1}{2}$.
Furthermore, we can assume it preserves the value of its input: seen as an element of  $\Netnotation{\ReLU}{1}{2}{2}$, it computes $x \mapsto (x,f(x))$ for $x \ge 0$. 
Furthermore, the network can be produced effectively and efficiently: given the $\ell_i$ by their rational coefficients, the resulting network has depth at most $k$ and is computable in time polynomial in the bit-size of this description; in particular, its rational weights have bit-size polynomial in that of the expression.\end{lemmarep}

\begin{proof}
Write $v_i$ for the value of the sub-expression rooted at level $i$,
so that $v_k=\ell_k(x)$ and $v_i=\tau_i(\ell_i(x),v_{i+1})$; thus
$f=v_1$. We construct, by downward induction on $i$, a width-$2$
network whose layers maintain the invariant: after the layer
processing level $i$, the first neuron carries $x$ (as $x=\ReLU(x)$
for $x\ge 0$) and the second neuron carries a value $r_i$ such that
\[
  v_i \;=\; m_i(x) + \sigma_i\, r_i
\]
for some rational affine form $m_i$ and sign $\sigma_i\in\{-1,+1\}$,
both known at construction time.

Base case: the input affine map initialises the pair $(x,0)$, and
$v_k=\ell_k(x)$ is realised with $m_k=\ell_k$, $r_k=0$,
$\sigma_k=+1$ (no layer needed). Inductive step: suppose the
invariant holds at level $i+1$. If $\tau_i=\min$, use the identity
$\min(a,b)=a-\ReLU(a-b)$:
\[
  r_i \;:=\; \ReLU\bigl(\ell_i(x)-v_{i+1}\bigr)
        \;=\; \ReLU\bigl(\ell_i(x)-m_{i+1}(x)-\sigma_{i+1}r_{i+1}\bigr),
\]
the $\ReLU$ of a rational affine form in $(x,r_{i+1})$, hence one
width-$2$ layer, the other neuron refreshing $x$; and
$v_i=\ell_i(x)-r_i$, i.e.\ $m_i=\ell_i$, $\sigma_i=-1$. If
$\tau_i=\max$, use $\max(a,b)=a+\ReLU(b-a)$ symmetrically:
$r_i:=\ReLU\bigl(m_{i+1}(x)+\sigma_{i+1}r_{i+1}-\ell_i(x)\bigr)$ and
$v_i=\ell_i(x)+r_i$. After processing level $1$, the output affine
map returns $(x,\,m_1(x)+\sigma_1 r_1)=(x,f(x))$. The construction
uses one layer per level, hence depth at most $k$, and every weight
is a rational combination of the coefficients of the $\ell_i$, of
bit-size polynomial in the input.
\end{proof}
We will need the following gadgets:

\begin{lemmarep}[Leading-digit extraction]\label{lem:lead}
There exists a ReLU neural network \(\NNlead:\mathbb{R}\to\mathbb{R}^2\) in \(\Netnotation{\ReLU}{1}{2}{2}\)
such that, for every \(x\in(0,1)\) written in base \(4\) as
$x=\sum_{i=1}^{\infty} a_i\,4^{-i},$ we have $\NNlead(x)=(x,\lead(x)),$
where \(\lead(x)=a_1\) whenever \(a_1\in\{1,3\}.\) Moreover,  if $a_1=0$, then  \(\lead(x)=1\).
\end{lemmarep}

\begin{proof}
The construction relies on a piecewise affine function that distinguishes the two admissible leading base-\(4\) digits. More precisely, if \(x=\sum_{i=1}^{\infty} a_i4^{-i}\), then \(x\in\left[\frac14,\frac12\right)\) when \(a_1=1\), and \(x\in\left[\frac34,1\right)\) when \(a_1=3\). Hence, it is enough to define a piecewise affine map \(\lead:[0,1]\to\mathbb{R}\) satisfying \(\lead(x)=1\) on \(\left[\frac14,\frac12\right)\) and \(\lead(x)=3\) on \(\left[\frac34,1\right)\). For instance, one may take
\[
\lead(x)=
\begin{cases}
1, & \text{if } 0\le x\le \frac12,\\
8x-3, & \text{if } \frac12\le x\le \frac34,\\
3, & \text{if } \frac34\le x\le 1.
\end{cases}
\]
This function is continuous, piecewise affine, and correctly recovers the leading digit on all encodings with first digit in \(\{1,3\}\). Such a function is realized by a width-\(2\), depth-\(2\) ReLU network. 
Figure~\ref{fig:lead} depicts the corresponding network.

\input{./fig-lead}

\end{proof}

\begin{lemmarep}[Base-$4$ digit shift]\label{lem:shift}
There exists a ReLU neural network \(\NNshift:\mathbb{R}\to\mathbb{R}^2\) in \(\Netnotation{\ReLU}{1}{2}{2}\) such that, for every \(x\in(0,1)\),
$\NNshift(x)=(x,\shift(x)),$
where
\(\shift:\mathbb{R}\to\mathbb{R}_{\ge 0}\)
satisfies
\(\shift(x)\in[0,1]\)
for every \(x\in(0,1)\), and
$\shift(x)=\sum_{i=2}^{\infty} a_i\,4^{-(i-1)}$ whenever
$x=\sum_{i=1}^{\infty} a_i\,4^{-i}$  and  $a_1\in\{1,3\}.$
\end{lemmarep}
\begin{proof}
It suffices to verify that $\shift(x)=\ReLU(4x-\lead(x))$. Indeed, if $x=\sum_{i=1}^{\infty} a_i\,4^{-i}$
with \(a_1\in\{1,3\}\), then by the definition of \(\lead\) we have
$\lead(x)=a_1$. Hence
$4x-\lead(x)=4\sum_{i=1}^{\infty} a_i\,4^{-i}-a_1 =\sum_{i=1}^{\infty} a_i\,4^{-(i-1)}-a_1=\sum_{i=2}^{\infty} a_i\,4^{-(i-1)}$.
Therefore
$\shift(x)=\sum_{i=2}^{\infty} a_i\,4^{-(i-1)}$,
as claimed.

Moreover, for \(x\in(0,1)\), the quantity \(\shift(x)\) belongs to \([0,1]\).
The neural network outputs both \(x\) and \(\shift(x)\), so it realizes the map
$x\longmapsto (x,\shift(x)). $
Since both the identity and the map \(x\mapsto 4x-\lead(x)\) are realized by a ReLU network of architecture
\(
\Netnotation{\ReLU}{1}{2}{2},
\)
this proves the claim (See Figure~\ref{fig:shift}).

\input{./fig-shift}
\end{proof}

\begin{lemmarep}[Conditional flag] \label{lem:test}
Let \(p\) be a value in \([0,4^{-1}] \cup [3\cdot 4^{-1},1]\), and let \(e\in\{0,1\}\). There exists a ReLU neural network \(\NNtest:\mathbb{R}^2\to\mathbb{R}\) in \(\Netnotation{\ReLU}{2}{1}{2}\)
that computes a function \(\test:\{0,1\}\times\mathbb{R}\to\mathbb{R}\)
satisfying
\[
\test(e,p)=
\begin{cases}
p-\frac34, & \text{if } e=1 \text{ and } p\ge \frac34,\\
p, & \text{otherwise.}
\end{cases}
\]
\end{lemmarep}

\begin{proof}
Let
$
f(e,p)=2e\,4^{-1}+p\,4^{-2},
$
and define \(g:\mathbb{R}\to\mathbb{R}\) by
\[
g(x)=
\begin{cases}
16x, & x \le \frac{1}{16},\\
\frac{8}{7}-\frac{16}{7}x, & x \le \frac{1}{2},\\
16x-8, & \frac{1}{2}\le x \le \frac{33}{64},\\
\frac{35}{8}-8x, & \frac{33}{64}\le x \le \frac{35}{64},\\
16x-\frac{35}{4}, & x \ge \frac{35}{64}.
\end{cases}
\]
We show that the composition \(g(f(e,p))\) matches the desired specification.
We distinguish cases (see Figure~\ref{fig:test}).

\begin{enumerate}
\item If \(e=0\), then $0\le f(e,p)\le f(0,1)<4^{-2}.$
Hence $g(f(0,p))=16f(0,p)=16\,(p\cdot 4^{-2})=p.$

\item If \(e=1\) and \(p\in[0,4^{-1}]\), then
$
f(e,p)=\frac12+p\cdot 4^{-2},
$
so
$
\frac12 \le f(e,p)\le \frac12+4^{-3}=\frac{33}{64}.
$
Therefore,
$
g(f(1,p))=16f(1,p)-8
=16\Bigl(\frac12+p\cdot 4^{-2}\Bigr)-8
=p.
$

\item If \(e=1\) and \(p\in[3\cdot 4^{-1},1]\), then
$
\frac12+3\cdot 4^{-3}\le f(e,p)\le \frac12+4^{-2},
$
that is,
$
\frac{35}{64}\le f(e,p)\le \frac{9}{16}.
$
Hence
$
g(f(1,p))
=
16\Bigl(\frac12+p\cdot 4^{-2}\Bigr)-\frac{35}{4}
=
(8+p)-\frac{35}{4}
=
p-\frac34.
$
\end{enumerate}

Thus \(g(f(e,p))=\test(e,p)\) in all cases.

The function $g$ admits the nested form
$g=\max\bigl(16x-\tfrac{35}{4},\,\min\bigl(\tfrac{35}{8}-8x,\,\max\bigl(16x-8,\,\min\bigl(\tfrac{8}{7}-\tfrac{16}{7}x,\,16x\bigr)\bigr)\bigr)\bigr)$,
as a case inspection over its pieces confirms, so Lemma~\ref{lem:CPAM} applies. Therefore, there exists a ReLU neural network
\(
\NN_g\in\Netnotation{\ReLU}{1}{1}{2}
\)
computing \(x\mapsto g(x)\) for \(x\ge 0\).

Now define \(\NNtest\) as follows: on input \((e,p)\), it first computes \(f(e,p)\), which can be implemented with a single ReLU neuron, and then applies \(\NN_g\) to obtain \(g(f(e,p))\). Hence
$\NNtest\in\Netnotation{\ReLU}{2}{1}{2},$
and \(\NNtest\) computes \(\test\), as claimed.

\input{./fig-test}

\end{proof}

\begin{lemmarep}[Counting the $3$'s in a base-$4$ expansion]
\label{lem:counting}
Fix $n\in\mathbb{N}$. There exists a neural network
$\NNCOUNT{n}\in\Netnotation{\ReLU}{2}{1}{3}$ such that, for every
$x=\sum_{i=1}^{n}a_i\,4^{-i}\in[0,1]$ with $a_i\in\{0,1,3\}$ for all
$1\le i\le n$, and every $p\in\mathbb{R}$,
\(
  \NNCOUNT{n}(x,p) \;=\; p + c,
\)
where $c=|\{i\in\{1,\dots,n\}:a_i=3\}|$.
\end{lemmarep}

\begin{proof}
Let $\eta := 4^{-(n+1)}$ 
and define the piecewise-affine
function $s\colon [0, 1] \to \mathbb{R}$ by
\[
  s(x) \;:=\;
  \begin{cases}
    0                               & x \in [0,\,\tfrac14 - \eta],\\
    \dfrac{x - (\tfrac14 - \eta)}{\eta} & x \in [\tfrac14 - \eta,\,\tfrac14],\\
    1                               & x \in [\tfrac14,\,\tfrac12 - \eta],\\
    1 + \dfrac{2(x - (\tfrac12 - \eta))}{\tfrac14 + \eta}
                                   & x \in [\tfrac12 - \eta,\,\tfrac34],\\
    3                               & x \in [\tfrac34,\,1].
  \end{cases}
\]
This function is continuous (the definitions match at the breakpoints)
and admits the nested form
$s=\min\bigl(3,\;\max\bigl(1+\tfrac{2(x-(1/2-\eta))}{1/4+\eta},\;
\min\bigl(1,\;\max\bigl(\tfrac{x-(1/4-\eta)}{\eta},\,0\bigr)\bigr)\bigr)\bigr)$,
hence by Lemma~\ref{lem:CPAM} it is realisable by a width-$2$ $\ReLU$
network.

\medskip
The network $\NNCOUNT{n}$ maintains state $(x, c)$ and implements:
\begin{enumerate}
\item $c \gets p$;
\item \textbf{for} $i = 1$ \textbf{to} $n$:
\begin{enumerate}
\item $c \gets c + \tfrac12\bigl(\lead(x) - 1\bigr)$;
\item $x \gets \ReLU\bigl(4x - s(x)\bigr)$;
\end{enumerate}
\item \textbf{return} $c$.
\end{enumerate}

Each step uses width-$2$ subnetworks ($\lead$ or $s$) while carrying the
other state variable on one additional neuron, giving overall width $3$.

\medskip
\noindent\textbf{Correctness.}
We prove by induction on $i$ that, after iteration $i$,
\[
  c = p + |\{j \le i : a_j = 3\}|
  \quad\text{and}\quad
  x = \sum_{j = i+1}^{n} a_j\,4^{-(j-i)}.
\]
The base case $i = 0$ is immediate.

For the inductive step, let $x_{i-1} = \sum_{j=i}^n a_j\,4^{-(j-i+1)}$ by
the inductive hypothesis. Since every $a_j \in \{0, 1, 3\}$, $x_{i-1}$
lies in exactly one of three regions, determined by $a_i$:
\begin{itemize}
\item If $a_i = 0$: $x_{i-1} \in [0,\,\sum_{j=i+1}^n 3\cdot 4^{-(j-i+1)}]
      \subseteq [0,\,\tfrac14 - 4^{-n}]$. Since $4^{-n} > \eta$,
      $x_{i-1} \in [0,\,\tfrac14 - \eta]$, so $s(x_{i-1}) = 0$ and
      $\lead(x_{i-1}) = 1$. Hence $c$ is unchanged (increment $0$), and
      $x \gets \ReLU(4x_{i-1}) = 4x_{i-1} = \sum_{j=i+1}^n a_j\,4^{-(j-i)}$.
\item If $a_i = 1$: $x_{i-1} \in [\tfrac14,\,\tfrac14 + \sum_{j=i+1}^n
      3\cdot 4^{-(j-i+1)}] \subseteq [\tfrac14,\,\tfrac12 - 4^{-n}]
      \subseteq [\tfrac14,\,\tfrac12 - \eta]$, so $s(x_{i-1}) = 1$
      and $\lead(x_{i-1}) = 1$. Hence $c$ is unchanged, and
      $x \gets \ReLU(4x_{i-1} - 1) = 4x_{i-1} - 1 = \sum_{j=i+1}^n
      a_j\,4^{-(j-i)}$.
\item If $a_i = 3$: $x_{i-1} \in [\tfrac34,\,\tfrac34 + \sum_{j=i+1}^n
      3\cdot 4^{-(j-i+1)}] \subseteq [\tfrac34,\,1]$, so $s(x_{i-1}) = 3$
      and $\lead(x_{i-1}) = 3$. Hence $c$ is incremented by
      $\tfrac12(3 - 1) = 1$, and $x \gets \ReLU(4x_{i-1} - 3)
      = 4x_{i-1} - 3 = \sum_{j=i+1}^n a_j\,4^{-(j-i)}$.
\end{itemize}
In every case the inductive invariant is preserved. After $n$ iterations,
$c = p + |\{j : a_j = 3\}|$, as claimed.
\end{proof}

\begin{lemmarep}[Truncation and $0$-digit filtering]\label{lem:filter}
\label{lem:trunck}
Let $n \in \mathbb{N}$. There exists a network
$\NNT{n}\in\Netnotation{\ReLU}{2}{2}{3}$ such that, for every $x \in (0,1)$ with
base-$4$ expansion $x = \sum_{i=1}^{\infty} a_i\,4^{-i}$, the output
$\NNT{n}(x,0) = (y, x_n)$ satisfies the following properties:
\begin{enumerate}
\item If $a_i \in \{1,3\}$ for all $1 \le i \le n$, then
      $y = \sum_{i=1}^n a_i\,4^{-i}$.
\item If $a_i \in \{0,1,3\}$ for all $1 \le i \le n$, then
      $y = \sum_{i=1}^n a'_i\,4^{-i}$ for some digits $a'_i \in \{1,3\}$.
      In particular, $a'_i \neq 0$ for all $1 \le i \le n$.
\end{enumerate}
\end{lemmarep}

\begin{proof}
We exploit refined properties of the functions $\lead$ and $\shift$
(Lemmas~\ref{lem:lead} and~\ref{lem:shift}; see also
Figures~\ref{fig:lead} and~~\ref{fig:shift}).
Recall that $\shift(x) = \ReLU(4x - \lead(x))$.
A direct inspection shows that if $a_1 = 0$ then $\lead(x) = 1$ and
$\shift(x) = 0$.

The network implements the following procedure, taking $(x, p)$ as input:
\begin{enumerate}
\item $y \leftarrow p$.
\item For $i = 1$ to $n$:
\begin{enumerate}
  \item $a \leftarrow \lead(x)$;
  \item $x \leftarrow \shift(x)$;
  \item $y \leftarrow \ReLU(y + a\,4^{-i})$.
\end{enumerate}
\item \textbf{return} $(y, x)$.
\end{enumerate}
The network has width $3$: it is assembled from width-$2$ subnetworks, and
at every step of the computation only one extra variable---either $y$ or
$x$, but never both simultaneously---needs to be stored alongside the
current subnetwork outputs.

\medskip
\noindent\textbf{Inductive invariant.}
For $i \ge 0$, let $(x_i, y_i)$ denote the values of $x$ and $y$ after
$i$ iterations of the loop, starting from $(x_0, p_0) = (x, 0)$.
We claim that, as long as $a_1,\dots,a_i \in \{1,3\}$,
\begin{equation}\label{eq:trunck:10}
x_i = \sum_{j=i+1}^{\infty} a_j\,4^{-(j-i)}
  \qquad\text{and}\qquad
  y_i = \sum_{j=1}^{i} a_j\,4^{-j}.
\end{equation}
The base case $i=0$ is immediate.  For the inductive step, assuming
Equation\eqref{eq:trunck:10} holds at step $i$ with $a_{i+1}\in\{1,3\}$, the
definitions of $\lead$ and $\shift$ give $\lead(x_i)=a_{i+1}$ and
$x_{i+1}=\sum_{j=i+2}^{\infty}a_j\,4^{-(j-i-1)}$, and then
$y_{i+1}=\ReLU(y_i+a_{i+1}\,4^{-(i+1)})=y_i+a_{i+1}\,4^{-(i+1)}$,
since $y_i + a_{i+1}\,4^{-(i+1)} > 0$.  This establishes the induction.

\medskip
\noindent\textbf{Point~(1).}
If $a_i \in \{1,3\}$ for all $i$, then Equation~\eqref{eq:trunck:10} applies at
every step.  After $n$ iterations we obtain
$x_n=\sum_{j=n+1}^{\infty}a_j\,4^{-(j-n)}$ and
$y_n = \sum_{j=1}^{n} a_j\,4^{-j}$, which is the claimed value of $y$.

\medskip
\noindent\textbf{Point~(2).}
Suppose $a_i \in \{0,1,3\}$ for all $i$.
If no digit equals $0$, Point~(1) applies directly.
Otherwise let $\ell \le n$ be the smallest index with $a_\ell = 0$.
Since $a_1,\dots,a_{\ell-1} \in \{1,3\}$, Equation~\eqref{eq:trunck:10} gives
\[
x_{\ell-1} = \sum_{j=\ell}^{\infty} a_j\,4^{-(j-\ell+1)},
  \qquad
  y_{\ell-1} = \sum_{j=1}^{\ell-1} a_j\,4^{-j}.
\]
Because $a_\ell = 0$, we have $\lead(x_{\ell-1}) = 1$ and
$\shift(x_{\ell-1}) = 0$, so $x_\ell = 0$.  For all subsequent iterations
$\ell' \ge \ell$, we have $x_{\ell'} = 0$ and hence $\lead(x_{\ell'}) = 1$
and $\shift(x_{\ell'}) = 0$.  The loop therefore accumulates
$a'_{\ell'} = 1$ for every $\ell' \ge \ell$, and $a'_j = a_j \in \{1,3\}$
for $j < \ell$.  Consequently, $y = \sum_{i=1}^{n} a'_i\,4^{-i}$ with
all digits $a'_i \in \{1,3\}$, which proves Point~(2).

\end{proof}

\section{Results for inputs restricted to $\onethreefouradic$}
\label{sec:restricted}

Throughout this section, all weights and biases are rational numbers
with polynomial encoding size; network evaluation is performed in
exact rational arithmetic. We show that, in width~$3$, the decision
problem ``does there exist an input on which the network outputs
zero?'' is already $\NP$-hard when inputs are restricted to
$\{1,3\}$-quadradic encodings of length $n$; the unrestricted version
with inputs in all of $[0,1]$ is reduced to this in
Section~\ref{sec:preprocessing} using the fractal preprocessing
gadget.

\begin{definition}[\textsc{Verification}${}_{\{1,3\}}$]
\label{def:restricted-nn-zero}
\hfill
\begin{itemize}
\item \emph{Input:} a $\ReLU$ neural network
  $\NN\colon[0,1]\to\mathbb{R}$ and an integer $n\in\mathbb{N}$.
\item \emph{Question:} does there exist
  $x=\sum_{i=1}^{n}a_i\,4^{-i}\in[0,1]$
  with $a_i\in\{1,3\}$ such that $\NN(x)=0$?
\end{itemize}
\end{definition}

\begin{theorem}\label{thm:np-complete-encoded-zero}\label{th:NP:restricted}
\textsc{Verification}${}_{\{1,3\}}$ is $\NP$-complete, already for
$\ReLU$ neural networks of width $3$.
\end{theorem}

\begin{proofsketch}
Hardness is by a
polynomial-time reduction from \textsc{3-Sat}: the network
$\NN_\varphi$ of Section~\ref{sec:reduction} composes $m$ clause
subnetworks of width~$2$ that scan the $n$ base-$4$ digits of the
input, evaluate clause $C_\ell$ via the $\test$ gadget of
Lemma~\ref{lem:test}, and leave a flag in $\{0,3/4\}$; carrying the
input register costs one extra neuron, hence width~$3$. Correctness
is Claims~\ref{claim:x} and~\ref{claim:assignment}.
\end{proofsketch}

\begin{remark}[Relation to~\cite{MFCS2025}]\label{rem:baseline}
The digit gadgets above build on~\cite{MFCS2025}, whose Theorem~23
states that the transition map of a Turing machine, acting on
configurations already encoded as pairs of $\{1,3\}$-quadradics, can
be implemented at width~$3$.  This suggests a second route to
Theorem~\ref{thm:np-complete-encoded-zero}: supply $x(\alpha)$ as an
affine initial configuration of a polynomial-time verifier with
$\varphi$ hard-wired in its finite control, and unroll the transition
map for its step bound. In this sense we regard the discrete statement
as a baseline. Making it work at width~$3$ is, however, where the care
lies. In a verification instance, the input register must remain
available while the digits are processed, a requirement that does not
arise in~\cite{MFCS2025}, where the transition map acts on the
configuration itself. Here it is guaranteed by the compilation
discipline of Lemma~\ref{lem:CPAM}, which preserves the input at no
extra width, whereas generic compilations of nested min/max
expressions give width~$4$~\cite[Remark~22]{MFCS2025}. The direct
reduction also makes the output explicit: by
Claim~\ref{claim:assignment} it records which clauses are satisfied,
and the companion results of Section~\ref{sec:variants} are read off
this construction. Its explicit Lipschitz bounds are, moreover, what
the lifting to~$[0,1]$ in Section~\ref{sec:unrestricted} requires.
\end{remark}

\paragraph*{Membership in $\NP$} 
A polynomial-size certificate is the sequence of digits
$(a_1,\dots,a_n)\in\{1,3\}^n$. From this sequence, one computes
$x=\sum a_i\,4^{-i}$ and evaluates $\NN(x)$ in polynomial time. $\NP$-hardness is
established by a polynomial-time reduction from \textsc{3-Sat} below.

\subsection{Reduction from \textsc{3-Sat}} \label{sec:reduction}
Let $\varphi = \bigwedge_{\ell=1}^m C_\ell$ be a 3-CNF formula over
variables $y_1,\dots,y_n$, with clauses
$C_\ell = L_{\ell,1} \vee L_{\ell,2} \vee L_{\ell,3}$, where each
$L_{\ell,j}$ is either some $y_i$ or its negation $\neg y_i$. We construct,
in polynomial time, a width-$3$ $\ReLU$ network
$\NN_\varphi\colon [0,1]\to\mathbb{R}$ with
\(
\varphi\text{ satisfiable}
\iff
\exists\,x\in\onethreefouradicn{n}\colon\NN_\varphi(x)=0.
\)

\paragraph*{Encoding.} An assignment $\alpha\in\{0,1\}^n$ is
encoded as
$x(\alpha) = \sum_{i=1}^n a_i\,4^{-i}\in\onethreefouradicn{n}$, with
$a_i=1$ if $\alpha(y_i)=1$ and $a_i=3$ otherwise. This defines a bijection
$\{0,1\}^n \leftrightarrow \onethreefouradicn{n}$, whose inverse is given by
$\alpha(y_i)=\ReLU((3-a_i)/2)$.

\paragraph*{Overall structure.}
The network operates on a two-register state $(x,p)\in[0,1]^2$,
with $p$ a cumulative clause flag initialised to $0$. The $m$
clauses are processed sequentially: for each $C_\ell$, a
subnetwork $\NN_{\varphi,\ell}$ maps $(x,p_{\ell-1})\mapsto(x,p_\ell)$.
The full reduction is the composition
\[
  \NN_\varphi \;=\; \pi_2\;\circ\;\NN_{\varphi,m}\;\circ\;\cdots\;\circ\;\NN_{\varphi,1}
  \quad\text{applied to }(x,0),
\]
where $\pi_2$ projects onto the second coordinate. The register $p_\ell$ is
designed so that $p_\ell=0$ if and only if $C_1,\dots,C_\ell$ are all satisfied by
the assignment encoded in~$x$; hence $\NN_\varphi(x(\alpha))=0$  if and only if
$\alpha$ satisfies $\varphi$. See 
Section~\ref{sec:correctness} for the formal proof.

\paragraph*{The clause subnetwork $\NN_{\varphi,\ell}$.}
Given $(x,p_{\ell-1})$, the subnetwork $\NN_{\varphi,\ell}$ proceeds in
two phases.
\emph{Phase~1} opens a fresh clause flag by
\(
  p_\ell \gets \min(3/4 + p_{\ell-1}/4,\,1),
\)
which sets $p_\ell$ to the ``unsatisfied'' value $3/4$ while
compressing the previous flag $p_{\ell-1}$ into $[0,\tfrac14]$ so
that it cannot interfere with subsequent updates.
\emph{Phase~2} scans the $n$ digits of~$x$ by composing  $n$ inner
blocks $\NN_{\varphi,\ell,1},\dots,\NN_{\varphi,\ell,n}$ with identical
layer structure (only $\ellit_{\ell,j}$ varies with~$j$): on input
$(x,p_\ell)$, each $\NN_{\varphi,\ell,j}$ performs the following operations:
\begin{enumerate}
\item extract and booleanise the leading digit:
      $a_j\gets\lead(x)$, $\;b_j\gets\ReLU(\tfrac{3-a_j}{2})$;
\item clear the flag when the literal $L_{\ell,j}$ associated  with $y_j$ in $C_\ell$ is satisfied:
      $p_\ell\gets\test(\ellit_{\ell,j}(b_j),\,p_\ell)$,
      where
      $\ellit_{\ell,j}(b)=\ReLU(b)$ if $y_j\in C_\ell$,
      $\ReLU(1-b)$ if $\neg y_j\in C_\ell$, and $0$ otherwise;
\item rotate the leading digit to the least significant position:
      $x\gets\shift(x)+4^{-n}a_j$.
\end{enumerate}
Since $\test$ subtracts $3/4$ at most once (Lemma~\ref{lem:test}),
later satisfying digits leave a cleared flag untouched; after the
$n$ inner blocks, the digits of~$x$ have been cyclically rotated
back to their initial configuration, so $\NN_{\varphi,\ell}$ returns
$(x,p_\ell)$ with the input register preserved. See
Figure~\ref{fig:NP:network-4}.

\paragraph*{Example.}
Take $n=2$ and $\varphi=C_1\wedge C_2$ with $C_1=y_1\vee\neg y_2$ and
$C_2=\neg y_1\vee y_2$ (clauses of size two, for readability), and
$\alpha=(1,0)$, so that $x(\alpha)=1\cdot4^{-1}+3\cdot4^{-2}=\tfrac{7}{16}$.
The registers evolve as follows:
\begin{center}\small
\begin{tabular}{lccccc}
block & $a_j$ & $b_j$ & $\ellit_{\ell,j}(b_j)$ & $x$ & $p$\\\hline
start & & & & $7/16$ & $0$\\
$C_1$, phase 1 & & & & $7/16$ & $3/4$\\
$C_1$, $j=1$ & $1$ & $1$ & $1$ & $13/16$ & $0$\\
$C_1$, $j=2$ & $3$ & $0$ & $1$ & $7/16$ & $0$\\
$C_2$, phase 1 & & & & $7/16$ & $3/4$\\
$C_2$, $j=1$ & $1$ & $1$ & $0$ & $13/16$ & $3/4$\\
$C_2$, $j=2$ & $3$ & $0$ & $0$ & $7/16$ & $3/4$
\end{tabular}
\end{center}
The flag of $C_1$ is cleared, that of $C_2$ is not, and
$\NN_\varphi(x(\alpha))=\tfrac34=c_2\,4^{-1}+c_1\,4^{-2}$ with
$c_1=0$, $c_2=3$, as in Claim~\ref{claim:assignment}. For the
satisfying assignment $(1,1)$, i.e.\ $x=\tfrac{5}{16}$, both flags are
cleared and the output is~$0$.

\input{./fig-NP-3}

\paragraph*{Width.} Each inner step uses a width-$2$ subnetwork
($\NNlead$, $\NNshift$ on $x$ alone, or $\NNtest$ on
$(\ellit_{\ell,j}(b_j),\,p_\ell)$) while the other register is carried
by one extra neuron; all remaining updates are affine. The total width is $3$,
the depth is $\mathcal{O}(mn)$, and all coefficients are rational of polynomial
encoding size; hence 
$\NN_\varphi\in\Netnotation{\ReLU}{1}{1}{3}$.

\subsection{Correctness of the reduction}
\label{sec:correctness}

The correctness rests on two invariants: a \emph{structural} one
(Claim~\ref{claim:x}) showing that the input register is preserved by
every clause subnetwork on valid encodings, and a \emph{semantic} one
(Claim~\ref{claim:assignment}) giving an explicit base-$4$ expansion
of $p_\ell$ whose digits record clause satisfaction. Together they
yield
$\NN_\varphi(x(\alpha)) = 0 \iff \alpha \text{ satisfies } \varphi$.

\begin{claimrep}[Preservation of the input register]\label{claim:x}
Let $x \in \onethreefouradicn{n}$. Then for every
$\ell \in \{1,\dots,m\}$ and every $p \in [0,1]$, there exists
$p' \in [0,1]$ such that
\(
\NN_{\varphi,\ell}(x,p) = (x, p').
\)
\end{claimrep}

\begin{proof}
Write $x = \sum_{j=1}^{n} a_j\,4^{-j}$ with $a_j \in \{1,3\}$, and fix a
clause index $\ell$. The clause subnetwork $\NN_{\varphi,\ell}$ applies $n$
inner blocks $\NN_{\varphi,\ell,1},\dots,\NN_{\varphi,\ell,n}$ to the input
register; we show that each block performs a cyclic left-shift of the
base-$4$ digits of $x$.

Consider the first block. By Lemmas~\ref{lem:lead} and~\ref{lem:shift},
applied to $x \in \onethreefouradicn{n}$, we have
$\lead(x) = a_1$ and $\shift(x) = \sum_{j=2}^{n} a_j\,4^{-(j-1)}$.
The subsequent affine update $x \mapsto x + 4^{-n}\,a_1$ therefore produces
\[
  x^{(1)} \;=\; \sum_{j=2}^{n} a_j\,4^{-(j-1)} \;+\; a_1\,4^{-n}
         \;=\; \sum_{i=1}^{n} a_{(i \bmod n) + 1}\,4^{-i},
\]
i.e.\ $x^{(1)}$ has the digits of $x$ cyclically shifted one position to
the left. In particular, $x^{(1)} \in \onethreefouradicn{n}$, so the same
analysis applies to the next block.

By induction, after $k$ blocks the first coordinate is
\[
  x^{(k)} \;=\; \sum_{j=k+1}^{n} a_j\,4^{-(j-k)}
            \;+\; \sum_{j=1}^{k} a_j\,4^{-(n-k+j)}
         \;=\; \sum_{i=1}^{n} a_{((i+k-1) \bmod n)+1}\,4^{-i},
\]
which is the $k$-fold cyclic shift of $x$. Taking $k = n$ recovers the
identity permutation, so $x^{(n)} = x$. Since the second coordinate depends
only on the digits read and on $p$, this yields
$\NN_{\varphi,\ell}(x,p) = (x, p')$ for some $p'$.
\end{proof}

We now turn to the second coordinate. The next claim gives a closed-form expression
for $p_\ell$ and, in particular, characterises when $p_m = 0$.

\begin{claimrep}[Clause-satisfaction readout]\label{claim:assignment}
Let $\alpha \in \{0,1\}^n$ be an assignment and let
$x(\alpha) \in \onethreefouradicn{n}$ be its encoding. Define
\[
  c_\ell \;=\; \begin{cases}
    0 & \text{if } \alpha \text{ satisfies } C_\ell,\\
    3 & \text{otherwise,}
  \end{cases}
  \qquad \ell \in \{1,\dots,m\}.
\]
Then $\NN_\varphi(x(\alpha)) = p_m$, where
\(
  p_m \;=\; \sum_{\ell=1}^{m} c_\ell\,4^{-(m+1-\ell)}.
\)
Equivalently, the base-$4$ expansion of $p_m$ is $0.c_m c_{m-1} \cdots c_1$,
read from most to least significant digit; in particular, $p_m = 0$ if and
only if $\alpha$ satisfies every clause of $\varphi$.
\end{claimrep}

\begin{proof}
Fix a clause $C_\ell$ and suppose that the input to $\NN_{\varphi,\ell}$ is
$(x(\alpha), p_{\ell-1})$ with
$p_{\ell-1} \in [0, \tfrac14]\cup[\tfrac34,1)$ (a property we verify
by induction below; it matches the precondition of
Lemma~\ref{lem:test}).

The flag-initialisation step sets
\(
  p_\ell \gets \min\!\bigl(\tfrac34 + \tfrac14\,p_{\ell-1},\, 1\bigr)
= \tfrac34 + \tfrac14\,p_{\ell-1} \in [\tfrac34, 1),
\)
so the $\min$ is inactive. As the inner loop scans the digits of $x(\alpha)$,
the register $p_\ell$ is updated exclusively through the $\test$ operation.
By construction, $\test$ subtracts $\tfrac34$ from $p_\ell$ precisely when
$\ellit_{\ell,j}(b_j) = 1$, which, by inspection of $\ellit_{\ell,j}$,
happens if and only if variable $y_j$ appears in $C_\ell$ under a polarity
satisfied by $\alpha$. Moreover, since $p_\ell$ drops into $[0, \tfrac14]$
as soon as the first such subtraction occurs, Lemma~\ref{lem:test}
guarantees that no further $\test$ call modifies the register. Hence:
\begin{itemize}
\item if $\alpha$ satisfies $C_\ell$, exactly one subtraction occurs and
      $p_\ell = \tfrac14\,p_{\ell-1} \in [0, \tfrac14]$;
\item if $\alpha$ does not satisfy $C_\ell$, no subtraction occurs and
      $p_\ell = \tfrac34 + \tfrac14\,p_{\ell-1}$.
\end{itemize}
Both cases can be written uniformly as
\begin{equation}\label{eq:p-recurrence}
  p_\ell \;=\; c_\ell\,4^{-1} \;+\; p_{\ell-1}\,4^{-1},
\end{equation}
with $c_\ell \in \{0, 3\}$ as defined in the claim. In either case
$p_\ell \in [0, \tfrac14]\cup[\tfrac34,1)$: it lies in $[0,\tfrac14]$
when $c_\ell = 0$ and in $[\tfrac34,1)$ when $c_\ell = 3$; since
$p_0 = 0$, the induction is well-founded.

Unfolding \eqref{eq:p-recurrence} with $p_0 = 0$ yields
\(
  p_\ell = \sum_{j=1}^{\ell} c_j\,4^{-(\ell+1-j)}
\)
by a straightforward induction: the base case $\ell = 1$ gives
$p_1 = c_1\,4^{-1}$, and the induction step is
\[
  p_\ell
  \;=\; c_\ell\,4^{-1} + \tfrac14 \sum_{j=1}^{\ell-1} c_j\,4^{-(\ell-j)}
  \;=\; \sum_{j=1}^{\ell} c_j\,4^{-(\ell+1-j)}.
\]
Setting $\ell = m$ gives the stated formula for $p_m$. The digits of $p_m$
lie in $\{0, 3\}$, hence form a valid base-$4$ expansion; the digit at
position $4^{-i}$ is $c_{m+1-i}$, so reading from most to least significant
yields $c_m, c_{m-1}, \dots, c_1$. Finally, $p_m = 0$ if and only if every
$c_\ell$ equals $0$, that is, if and only if $\alpha$ satisfies every clause
of $\varphi$.
\end{proof}

\subsection{NP-hardness of \textsc{Verification}${}_{\{1,3\}}$}
\label{sec:np-hardness}

With the invariants from the previous subsection in hand, the two directions
of the reduction follow immediately.

\begin{claimrep}[Completeness]\label{claim:satis}
If $\varphi$ is satisfiable, then there exists
$x \in \onethreefouradicn{n}$ such that $\NN_\varphi(x) = 0$.
\end{claimrep}

\begin{proof}
Let $\alpha$ be a satisfying assignment and let
$x(\alpha) = \sum_{i=1}^{n} a_i\,4^{-i} \in \onethreefouradicn{n}$
be its encoding. By Claim~\ref{claim:assignment},
\(
  \NN_\varphi(x(\alpha)) \;=\; \sum_{\ell=1}^{m} c_\ell\,4^{-(m+1-\ell)},
\)
with $c_\ell = 0$ whenever $\alpha$ satisfies $C_\ell$. Since $\alpha$
satisfies every clause, every $c_\ell$ vanishes, and therefore
$\NN_\varphi(x(\alpha)) = 0$.
\end{proof}

\begin{claimrep}[Soundness]\label{claim:satisconv}
If there exists $x \in \onethreefouradicn{n}$ such that $\NN_\varphi(x) = 0$,
then $\varphi$ is satisfiable.
\end{claimrep}

\begin{proof}
Let $x = \sum_{i=1}^{n} a_i\,4^{-i} \in \onethreefouradicn{n}$ satisfy
$\NN_\varphi(x) = 0$, and define an assignment $\alpha \in \{0,1\}^n$ by
\[
\alpha(y_i) \;=\;\ReLU\Bigl( \frac{3 - a_i}{2}\Bigr) \;=\;
  \begin{cases}
    1 & \text{if } a_i = 1,\\
    0 & \text{if } a_i = 3.
  \end{cases}
\]
By construction of the encoding, $x$ coincides with $x(\alpha)$, so
Claim~\ref{claim:assignment} applies and yields
\[
  0 \;=\; \NN_\varphi(x) \;=\; \NN_\varphi(x(\alpha))
    \;=\; \sum_{\ell=1}^{m} c_\ell\,4^{-(m+1-\ell)}.
\]
Each $c_\ell \in \{0,3\}$ is non-negative, so every $c_\ell$ must vanish;
by the definition of $c_\ell$, this means that $\alpha$ satisfies every
clause of $\varphi$. Hence $\varphi$ is satisfiable.
\end{proof}

Combining the two claims, the map $\varphi\mapsto\NN_\varphi$ is a
polynomial-time many-one reduction from \textsc{3-Sat} to
\textsc{Verification}${}_{\{1,3\}}$, and together with $\NP$-membership
this proves Theorem~\ref{thm:np-complete-encoded-zero}.
 
 \subsection{Variants and consequences}
\label{sec:variants}

The reduction of Section~\ref{sec:reduction} is parsimonious:
composing $\NN_\varphi$ with the counting gadget of
Lemma~\ref{lem:counting} yields a width-$3$ network whose value at
$x(\alpha)$ is exactly the number of clauses left unsatisfied by
$\alpha$. A single construction therefore transfers the 
complexity landscape of \textsc{3-Sat} to width-$3$ $\ReLU$
verification, across decision, counting, and approximation: an
$\NP/\coNP$-complete decision pair (Theorem~\ref{thm:np-conp}), a
$\coNP$-complete uniform-distance problem
(Corollary~\ref{cor:uniform}), a $\#\P$-complete counting problem
(Theorem~\ref{thm:sharp-p}), a $\PP$-complete majority problem
(Theorem~\ref{thm:pp}), and a \textsc{Max-3Sat} inapproximability
statement (Theorem~\ref{thm:inapprox}).  All results are
stated on $\onethreefouradicn{n}$; extensions to $[0,1]$ are given in
Section~\ref{sec:unrestricted}.

\medskip
\noindent\textbf{The counting network.}
Fix a 3-CNF formula $\varphi$ with $m$ clauses over $n$ variables. Define
\[
  \NN_{\varphi,D}(x) \;:=\; \NNCOUNT{m}\bigl(\NN_\varphi(x),\,0\bigr),
\]
where $\NN_\varphi$ is the network from Section~\ref{sec:reduction} and
$\NNCOUNT{m}$ is the counting gadget from Lemma~\ref{lem:counting} applied
to the $m$-digit base-$4$ expansion of its first argument.

\begin{lemmarep}[Master construction]\label{lem:master}
The network $\NN_{\varphi,D}\colon[0,1]\to\mathbb{R}$ lies in
$\Netnotation{\ReLU}{1}{1}{3}$, has depth $\mathcal{O}(mn)$, and is
computable in polynomial time from~$\varphi$. For every
$\alpha\in\{0,1\}^n$ with encoding
$x(\alpha)\in\onethreefouradicn{n}$,
\[
  \NN_{\varphi,D}(x(\alpha))
  \;=\;
  \bigl|\{\ell\in\{1,\dots,m\}: \alpha\text{ does not satisfy } C_\ell\}\bigr|
  \;\in\;\{0,1,\dots,m\}.
\]
\end{lemmarep}

\begin{proof}
By Claim~\ref{claim:assignment}, the base-$4$ expansion of
$\NN_\varphi(x(\alpha))$ has $m$ digits $c_1,\dots,c_m\in\{0,3\}$ with
$c_\ell=3$ iff $\alpha$ does not satisfy $C_\ell$. Since
$\{0,3\}\subset\{0,1,3\}$, Lemma~\ref{lem:counting} applies and counts the
$3$-digits of $\NN_\varphi(x(\alpha))$, which is precisely the number of
unsatisfied clauses. The width and depth bounds follow by concatenating
the two constructions.
\end{proof}

Equivalently, $\NN_{\varphi,D}(x(\alpha)) = m -
\operatorname{val}_\varphi(\alpha)$, where
$\operatorname{val}_\varphi(\alpha)$ is the number of clauses of $\varphi$
satisfied by $\alpha$---the objective function of \textsc{Max-3Sat}.
Decision, counting, majority, and approximation problems concerning
$\NN_{\varphi,D}$ on $\onethreefouradicn{n}$ are therefore in direct
correspondence with the analogous problems on~$\varphi$.

\paragraph*{$\coNP$-complete variants.}
The first pair of corollaries concerns existence and universality of
zeros.

\begin{definition}[\textsc{Universal-Gap}${}_{\{1,3\}}$]\label{def:exist-univ}
 the \emph{input} is
\begin{itemize}
\item \emph{Input:}  a $\ReLU$ neural network
$\NN\colon[0,1]\to\mathbb{R}$ and an integer $n\in\mathbb{N}$.
\item \emph{Question}: is
$\NN(x)\ge 1$ for every $x\in\onethreefouradicn{n}$?
\end{itemize}
\end{definition}

\begin{mytheoremrep}\label{thm:np-conp}
\textsc{Universal-Gap}${}_{\{1,3\}}$ is $\coNP$-complete, already for
$\ReLU$ neural networks of width~$3$.
\end{mytheoremrep}

\begin{proof}
For \textsc{Universal-Gap}${}_{\{1,3\}}$, membership in $\coNP$ is
witnessed by a digit sequence $(a_1,\dots,a_n)\in\{1,3\}^n$ for which
the rational $\NN\bigl(\sum_i a_i 4^{-i}\bigr)$, computable exactly in
polynomial time, is strictly less than~$1$. For $\coNP$-hardness, we
reduce from \textsc{3-Unsat}. By Lemma~\ref{lem:master},
\[
  \varphi\text{ is unsatisfiable}
  \;\Longleftrightarrow\;
  \forall\alpha\in\{0,1\}^n,\ \NN_{\varphi,D}(x(\alpha))\ge 1
  \;\Longleftrightarrow\;
  \forall x\in\onethreefouradicn{n},\ \NN_{\varphi,D}(x)\ge 1,
\]
so $\varphi\mapsto\NN_{\varphi,D}$ is a polynomial-time many-one
reduction from \textsc{3-Unsat} to
\textsc{Universal-Gap}${}_{\{1,3\}}$.
\end{proof}

The integer gap $\{0\}\cup[1,+\infty)$ in the output of $\NN_{\varphi,D}$
on $\onethreefouradicn{n}$ makes a second $\coNP$-hardness result
immediate: deciding whether a network is uniformly close to zero on
$\onethreefouradicn{n}$ is \textsc{Universal-Gap}${}_{\{1,3\}}$ in
disguise.

\begin{definition}[\textsc{Dyadic-Uniform-Zero}]
\hfill
\begin{itemize}
\item \emph{Input:} a $\ReLU$ network $\NN\colon[0,1]\to\mathbb{R}$, an
integer $n$, and a rational threshold $\tau>0$.
\item \emph{Question:} is
$\displaystyle\sup_{x\in\onethreefouradicn{n}}|\NN(x)|\le\tau$?
\end{itemize}
\end{definition}

\begin{corollaryrep}\label{cor:uniform}
\textsc{Dyadic-Uniform-Zero} is $\coNP$-complete, already for $\ReLU$
networks of width $3$ and threshold $\tau=1/2$.
\end{corollaryrep}

\begin{proof}
Membership in $\coNP$: a NO certificate is an
$x\in\onethreefouradicn{n}$ with $|\NN(x)|>\tau$. For hardness, define
$\widetilde{\NN}_\varphi(x):=\ReLU\bigl(1-\NN_{\varphi,D}(x)\bigr)$, which
remains a width-$3$ $\ReLU$ network. On $\onethreefouradicn{n}$,
\[
  \widetilde{\NN}_\varphi(x(\alpha))
  \;=\;
  \begin{cases} 1 & \text{if $\alpha$ satisfies $\varphi$,}\\
                0 & \text{otherwise.} \end{cases}
\]
Hence $\sup_{x\in\onethreefouradicn{n}}\widetilde{\NN}_\varphi(x)
= 0$ if $\varphi$ is unsatisfiable and $= 1$ otherwise, so the instance
$(\widetilde{\NN}_\varphi,n,1/2)$ is a YES instance of
\textsc{Dyadic-Uniform-Zero} iff $\varphi$ is unsatisfiable.
\end{proof}

\paragraph*{Counting and majority.}
Since $\alpha\mapsto x(\alpha)$ is a bijection between $\{0,1\}^n$ and
$\onethreefouradicn{n}$, zeros of $\NN_{\varphi,D}$ on
$\onethreefouradicn{n}$ are in bijection with satisfying assignments of
$\varphi$. This parsimonious correspondence transports hardness from
$\#\textsc{3-Sat}$ and \textsc{Maj-3Sat}.

\begin{definition}[\textsc{Counting-Zeros}${}_{\{1,3\}}$]
\hfill
\begin{itemize}
\item \emph{Input:} a $\ReLU$ network $\NN$ and an integer $n$.
\item \emph{Output:} the cardinality
$\bigl|\{x\in\onethreefouradicn{n}:\NN(x)=0\}\bigr|$.
\end{itemize}
\end{definition}

\begin{mytheoremrep}\label{thm:sharp-p}
\textsc{Counting-Zeros}${}_{\{1,3\}}$ is $\#\P$-complete, already for
$\ReLU$ networks of width~$3$.
\end{mytheoremrep}

\begin{proof}
Membership in $\#\P$: a nondeterministic polynomial-time machine guesses
$(a_1,\dots,a_n)\in\{1,3\}^n$ and accepts iff
$\NN\bigl(\sum_i a_i 4^{-i}\bigr)=0$; its accepting paths are in
bijection with the zeros of $\NN$ on $\onethreefouradicn{n}$.

Hardness: by Lemma~\ref{lem:master}, $\NN_{\varphi,D}(x(\alpha))=0$ iff
$\alpha$ satisfies $\varphi$, and $\alpha\mapsto x(\alpha)$ is a bijection
between $\{0,1\}^n$ and $\onethreefouradicn{n}$. Hence
$\bigl|\{x\in\onethreefouradicn{n}:\NN_{\varphi,D}(x)=0\}\bigr|
= \#\varphi$, the number of satisfying assignments. This is a
parsimonious reduction from $\#\textsc{3-Sat}$, which is $\#\P$-complete.
\end{proof}

\begin{definition}[\textsc{Majority-Zero}${}_{\{1,3\}}$]
\hfill
\begin{itemize}
\item \emph{Input:} a $\ReLU$ network $\NN$ and an integer $n$.
\item \emph{Question:} is
$\bigl|\{x\in\onethreefouradicn{n}:\NN(x)=0\}\bigr|>2^{n-1}$?
\end{itemize}
\end{definition}

\begin{mytheoremrep}\label{thm:pp}
\textsc{Majority-Zero}${}_{\{1,3\}}$ is $\PP$-complete, already for
$\ReLU$ networks of width~$3$.
\end{mytheoremrep}

\begin{proof}
Membership: the nondeterministic machine of the previous proof has
exactly $2^n$ computation paths, one per element of
$\onethreefouradicn{n}$, and accepts on a strict majority of them iff the
instance is a YES instance. Hardness: by the same parsimonious reduction,
$\bigl|\{x\in\onethreefouradicn{n}:\NN_{\varphi,D}(x)=0\}\bigr|=\#\varphi$,
so the instance is YES iff a strict majority of assignments satisfies
$\varphi$, i.e.\ iff $\varphi\in\textsc{Majority-3Sat}$, which is
$\PP$-complete.
\end{proof}

\paragraph*{Inapproximability of the minimum output.}
The master construction also inherits the optimal inapproximability of
\textsc{Max-3Sat} due to Håstad~\cite{hastad2001some}.

\begin{mytheoremrep}[Gap inapproximability]\label{thm:inapprox}
For every constant $\eta>0$, the following promise problem is $\NP$-hard:
given a $\ReLU$ network $\NN\in\Netnotation{\ReLU}{1}{1}{3}$ and an
integer $n$ (together with a parameter $m$ encoded with $\NN$),
distinguish
\begin{itemize}
\item \emph{YES instances:}
$\displaystyle\min_{x\in\onethreefouradicn{n}}\NN(x)=0$;
\item \emph{NO instances:}
$\displaystyle\min_{x\in\onethreefouradicn{n}}\NN(x)\ge
\Bigl(\tfrac{1}{8}-\eta\Bigr)\,m$.
\end{itemize}
In particular, no polynomial-time algorithm approximates
$\min_{x\in\onethreefouradicn{n}}\NN(x)$ within additive error
$(\tfrac{1}{16}-\eta)\,m$, unless $\P=\NP$.
\end{mytheoremrep}

\begin{proof}
Håstad's theorem~\cite{hastad2001some} implies that, for every constant
$\eta>0$, distinguishing satisfiable 3-CNF formulas from those in which
at most a fraction $7/8+\eta$ of the clauses can be simultaneously
satisfied is $\NP$-hard. Apply Lemma~\ref{lem:master}. In the YES case
(satisfiable $\varphi$), some $\alpha$ leaves no clause unsatisfied, so
$\min_{x\in\onethreefouradicn{n}}\NN_{\varphi,D}(x)=0$. In the NO case,
every $\alpha$ leaves at least $(1/8-\eta)\,m$ clauses unsatisfied, so
$\min_{x\in\onethreefouradicn{n}}\NN_{\varphi,D}(x)\ge(1/8-\eta)\,m$.

The additive-approximation statement is a direct consequence: any
polynomial-time algorithm producing $r$ with
$|r-\min_x\NN(x)|<(1/16-\eta)\,m$ would output $r<(1/16-\eta)\,m$ on YES
instances and $r>(1/16-\eta)\,m$ on NO instances, hence decide the gap
problem.
\end{proof}

\section{The fractal gadget: Filtering from $\mathbb{R}$ to $\onethreefouradic$}
\label{sec:filtering} \label{sec:fractal}\label{sec:preprocessing}

The previous section reduces \textsc{3-Sat} to a width-$3$ zero test
\emph{on the discrete encoding set} $\onethreefouradicn{n}$. To
lift this reduction to arbitrary real inputs in $[0,1]$, we need
a continuous $\ReLU$ gadget that recognises membership in
$\onethreefouradicn{n}$: it should vanish, or approximately vanish, on its $2^n$
elements while remaining bounded away from zero elsewhere. The
construction of this section provides such a gadget as the $n$-th
iterate $F_\varepsilon^{[n]}$ of a single width-$2$ piecewise-affine
self-map. The zero set of $F_\varepsilon^{[n]}$ is (approximately)
the $n$-th level of a Cantor-like fractal of~$[0,1]$, and
$F_\varepsilon$ itself is the refinement operator of that
fractal --- whence the name \emph{fractal gadget}. This is more than
terminology. The self-similarity of the encoding set is exactly what
keeps the width constant: a single fixed map $F_\varepsilon$, reused at
every iteration, refines the construction by one level, so the same
width-$2$ gadget certifies membership at level $n$ for all $n$
simultaneously. Recognising a $2^n$-point set with polynomially many
parameters forces massive reuse of the same affine data;
self-similarity is exactly the property that lets a single fixed-size
block, iterated $n$ times, perform this reuse, decoupling the
cardinality $2^n$ of the search space from the two units of width
spent recognising it. The finite-level digit filtering in the
$n$-th stage of this single refinement rule, not a separate ad hoc
construction for each level.

\paragraph*{Towards inputs in $[0,1]$.}
Section~\ref{sec:unrestricted} composes the gadget
$F_\varepsilon^{[n]}$ of Equation~\eqref{eq:Feps} with the filter of
Lemma~\ref{lem:filter} to lift the single-quantifier results above to
inputs in $[0,1]$, at the cost of a robustness parameter
$\varepsilon=2^{-\mathcal{O}(m(m+n))}$ and a preprocessing stage of
depth $\mathcal{O}(n)$.

The gadget is built in two steps. Subsection~\ref{sec:dis}
analyses an idealised discontinuous map $F\colon[0,1]\to[0,1]$
whose $n$-fold iterate vanishes exactly on $\onethreefouradicn{n}$;
this map is not $\ReLU$-realisable (all $\ReLU$-computable functions
are continuous), but it captures the combinatorial content of the gadget
in its cleanest form. Subsection~\ref{sec:ccont} replaces $F$ with a
continuous, width-$2$ $\ReLU$-realisable approximation
$F_\varepsilon$ controlled by a single parameter $\varepsilon>0$.

\subsection{Ideal discontinuous $F$}
\label{sec:dis}

Define $F:[0,1]\to[0,1]$ by Equation~\eqref{eq:F}.

\begin{figure}[h]
\centering
\begin{minipage}{0.42\textwidth}
\begin{equation}\label{eq:F}
F(x)=
\begin{cases}
1,      & 0\le x<\tfrac14,\\
4x-1,   & \tfrac14\le x<\tfrac12,\\
1,      & \tfrac12\le x<\tfrac34,\\
4x-3,   & \tfrac34\le x\le 1.
\end{cases}
\end{equation}
\end{minipage}
\hfill
\begin{minipage}{0.52\textwidth}
\centering
\begin{tikzpicture}[scale=2.5]
  \draw[->] (-0.02,0) -- (1.08,0) node[right] {$x$};
  \draw[->] (0,-0.02) -- (0,1.08) node[above] {$F(x)$};
  \foreach \x/\label in {0/0,0.25/$\tfrac14$,0.5/$\tfrac12$,0.75/$\tfrac34$,1/1}
    \draw (\x,0.01) -- (\x,-0.01) node[below=2pt] {\label};
  \foreach \y/\label in {1/1}
    \draw (0.01,\y) -- (-0.01,\y) node[left=2pt] {\label};
  \draw[dashed,gray] (0,1) -- (1,1);
  \draw[thick]
    (0,1) -- (0.25,1)
    (0.25,0) -- (0.5,1)
    (0.5,1) -- (0.75,1)
    (0.75,0) -- (1,1);
  \foreach \p in {(0,1),(0.25,1),(0.25,0),(0.5,1),(0.75,1),(0.75,0),(1,1)}
    \fill \p circle (0.3pt);
\end{tikzpicture}
\end{minipage}
\end{figure}

\begin{remark}[Important idea: $F$ acts as a fractal gadget]
When applied $n$ times, $F$ vanishes precisely on the $2^n$ elements of
$\onethreefouradicn{n}$ and nowhere else in $[0,1]$; Thus, $0$ is an
``accept'' signal for membership in the length-$n$ prefix of the
Cantor-like encoding set. The zero sets of the iterates
$F^{[0]}\subset F^{[1]}\subset\cdots$ form a nested family whose union
is dense in $\overline{\onethreefouradic}\cap[0,1]$
(Def.~\ref{def:quadratic}); each iterate refines the previous level
of the Cantor structure, replacing every element of
$\onethreefouradicn{n}$ with two elements of
$\onethreefouradicn{n+1}$. Figure~\ref{fig:F-iterates} illustrates this refinement for $F^{[2]}$ and $F^{[3]}$.
This self-similarity justifies the name
\emph{fractal gadget} for $F$ and its continuous approximation
(Proposition~\ref{lem:F-shift}).
\end{remark}
\begin{figure}[h]
\centering
\begin{minipage}{0.48\textwidth}
\centering
\begin{tikzpicture}[scale=2.5]
  \draw[->] (-0.02,0) -- (1.04,0) node[right] {$x$};
  \draw[->] (0,-0.02) -- (0,1.08) node[above] {$F^{[2]}(x)$};
  \draw[dashed,gray] (0,1) -- (1,1);
  \draw[thick]
    (0,1) -- (5/16,1)
    (6/16,1) -- (7/16,1)
    (8/16,1) -- (13/16,1)
    (14/16,1) -- (15/16,1);
  \foreach \a in {5,7,13,15}{
    \draw[thick] (\a/16,0) -- ({(\a+1)/16},1);
  }
  \foreach \a in {5,7,13,15}{
    \fill (\a/16,0) circle (0.25pt);
  }
  \draw (0.01,1) -- (-0.01,1) node[left] {$1$};
  \foreach \a/\lab in {5/{\frac5{16}},7/{\frac7{16}},13/{\frac{13}{16}},15/{\frac{15}{16}}}{
    \draw (\a/16,0.01) -- (\a/16,-0.01)
      node[below,font=\scriptsize] {$\lab$};
  }
\end{tikzpicture}
\end{minipage}
\hfill
\begin{minipage}{0.48\textwidth}
\centering
\begin{tikzpicture}[scale=2.5]
  \draw[->] (-0.02,0) -- (1.04,0) node[right] {$x$};
  \draw[->] (0,-0.02) -- (0,1.08) node[above] {$F^{[3]}(x)$};
  \draw[dashed,gray] (0,1) -- (1,1);
  \draw[thick]
    (0,1) -- (21/64,1)
    (22/64,1) -- (23/64,1)
    (24/64,1) -- (29/64,1)
    (30/64,1) -- (31/64,1)
    (32/64,1) -- (53/64,1)
    (54/64,1) -- (55/64,1)
    (56/64,1) -- (61/64,1)
    (62/64,1) -- (63/64,1);
  \foreach \a in {21,23,29,31,53,55,61,63}{
    \draw[thick] (\a/64,0) -- ({(\a+1)/64},1);
  }
  \foreach \a in {21,23,29,31,53,55,61,63}{
    \fill (\a/64,0) circle (0.25pt);
  }
  \draw (0.01,1) -- (-0.01,1) node[left] {$1$};
  \foreach \x/\lab in {0.25/{\frac14},0.5/{\frac12},0.75/{\frac34},1/1}{
    \draw (\x,0.01) -- (\x,-0.01) node[below,font=\scriptsize] {$\lab$};
  }
\end{tikzpicture}
\end{minipage}
\caption{The iterates $F^{[2]}$ (left) and $F^{[3]}$ (right) of the ideal
map $F$ of Equation~\eqref{eq:F}. The iterate $F^{[n]}$ equals $1$ outside
$2^n$ cylinders of length $4^{-n}$, on each of which it is affine of slope
$4^n$ (Lemma~\ref{lem:affine-cylinder}). Its zeros (dots) are exactly the
$2^n$ elements of $\onethreefouradicn{n}$
(Proposition~\ref{lem:F-shift}): $\tfrac{5}{16},\tfrac{7}{16},
\tfrac{13}{16},\tfrac{15}{16}$ for $n=2$. Passing from $F^{[2]}$ to
$F^{[3]}$, each zero $z$ splits into the two zeros $z+4^{-3}$ and
$z+3\cdot4^{-3}$, for instance $\tfrac{5}{16}$ into $\tfrac{21}{64}$ and
$\tfrac{23}{64}$.}
\label{fig:F-iterates}
\end{figure}

We now go to details: 
The map $F$ has two ascending pieces of slope $4$ and two plateaus of
value $1$, with jump discontinuities at $x=\tfrac14$ and $x=\tfrac34$.
When the leading base-$4$ digit of $x$ lies in $\{1,3\}$, the branch
$4x-a_1$ shifts the remaining digits one position to the left;
otherwise $F$ sends $x$ to the absorbing state $1$.

The two lemmas below make this intuition precise. The first describes the
action of $F^{[t]}$ on prefix cylinders; the second characterises the zero
set of $F^{[n]}$.

\begin{lemmarep}[Digit shift and affine action on cylinders]
\label{lem:affine-cylinder}
Let $F:[0,1]\to[0,1]$ be the map defined by Equation~\eqref{eq:F}, and for $n\ge 1$
let $F^{[n]}$ denote its $n$-fold iterate. Write every $x\in(0,1)$ in
base~$4$ as $x=\sum_{j=1}^{\infty} a_j\,4^{-j}$.
\begin{enumerate}
\item (\emph{Digit shift.}) If $a_1,\dots,a_n\in\{1,3\}$, then
\(
  F^{[n]}(x) \;=\; \sum_{j=n+1}^{\infty} a_j\,4^{-(j-n)}.
\)
\item (\emph{Affine action on cylinders.}) For every $t\ge 0$ and every
$z=\sum_{j=1}^{t} a_j\,4^{-j}\in\onethreefouradicn{t}$ (with the
convention $\onethreefouradicn{0}=\{0\}$ and $F^{[0]}=\mathrm{id}$), the
restriction of $F^{[t]}$ to the cylinder $[z,\,z+4^{-t})$ is affine of slope
$4^t$ and has a zero at $z$:
\(
  F^{[t]}(x) \;=\; 4^t(x-z), \qquad x\in[z,\,z+4^{-t}).
\)
\end{enumerate}
\end{lemmarep}

\begin{proof}
We first prove part~(1) for $n=1$ and then iterate.
Suppose $a_1=1$, so $x\in[\tfrac14,\tfrac12]$ (with equality at the right
endpoint only if every $a_j$, $j\ge 2$, equals~$3$). On this interval the
formula $F(x)=4x-1$ applies, giving
\[
  F(x) \;=\; 4\sum_{j=1}^{\infty} a_j\,4^{-j} - 1
        \;=\; a_1 + \sum_{j=2}^{\infty} a_j\,4^{-(j-1)} - 1
        \;=\; \sum_{j=2}^{\infty} a_j\,4^{-(j-1)}.
\]
The case $a_1=3$ is symmetric: $x\in[\tfrac34,1]$ and
$F(x)=4x-3=\sum_{j=2}^{\infty} a_j\,4^{-(j-1)}$. Hence, whenever
$a_1\in\{1,3\}$, the map $F$ strips the leading digit. Iterating this
identity $n$ times---which is possible since, at each step, the new
leading digit is still in $\{1,3\}$ by hypothesis---gives the statement
of~(1).

Part~(2) follows directly. Write $x=z+4^{-t}u$ with $u\in[0,1)$, and let
$u=\sum_{k=1}^{\infty} b_k\,4^{-k}$ be its base-$4$ expansion. Then the
first $t$ digits of $x$ are $a_1,\dots,a_t\in\{1,3\}$, followed by
$b_1,b_2,\dots$ at positions $t+1,t+2,\dots$. Applying~(1) with $n=t$,
\[
  F^{[t]}(x) \;=\; \sum_{k=1}^{\infty} b_k\,4^{-k} \;=\; u \;=\; 4^t(x-z),
\]
as claimed.
\end{proof}

We can then formalize that $F$ provides a \emph{test} for
membership in $\onethreefouradicn{n}$:

\begin{propositionrep}[Zero characterisation of $F^{[n]}$]
\label{lem:F-shift}\label{lem:machin}
Let $F:[0,1]\to[0,1]$ be the map defined by Equation~\eqref{eq:F}. For every $n\ge 1$
and every $x\in[0,1]$,
\[
  F^{[n]}(x) \;=\; 0
  \quad\Longleftrightarrow\quad
  x\in\onethreefouradicn{n}.
\]
In particular, $F^{[n]}$ has exactly $2^n$ zeros in $[0,1]$, one per
digit sequence $(a_1,\dots,a_n)\in\{1,3\}^n$.
\end{propositionrep}

\begin{proof}
($\Leftarrow$) If $x=\sum_{j=1}^{n}a_j\,4^{-j}\in\onethreefouradicn{n}$,
then the base-$4$ expansion of $x$ has digits $a_1,\dots,a_n\in\{1,3\}$
followed by zeros. Lemma~\ref{lem:affine-cylinder}(1) then gives
$F^{[n]}(x)=\sum_{j=n+1}^{\infty}0\cdot 4^{-(j-n)}=0$.

($\Rightarrow$) We argue by induction on $n$. For $n=1$, inspecting the
four pieces of $F$ shows that $F(x)=0$ requires either $x=\tfrac14$ (from
the piece $F(x)=4x-1$) or $x=\tfrac34$ (from $F(x)=4x-3$), and these are
precisely the elements of $\onethreefouradicn{1}=\{\tfrac14,\tfrac34\}$.

Assume $n\ge 2$ and $F^{[n]}(x)=0$. Then
$F(F^{[n-1]}(x))=0$, so $F^{[n-1]}(x)\in\{\tfrac14,\tfrac34\}$ by the base
case. In particular, $F^{[n-1]}(x)\neq 1$; but $F$ maps the two plateau
intervals $[0,\tfrac14)\cup[\tfrac12,\tfrac34)$ to the constant value $1$,
so $F^{[n-2]}(x)$ must avoid these intervals. Hence
$F^{[n-2]}(x)\in[\tfrac14,\tfrac12)\cup[\tfrac34,1]$, which means that the
leading base-$4$ digit of $F^{[n-2]}(x)$ lies in $\{1,3\}$. Iterating this
argument backwards, we obtain digits $a_1,\dots,a_n\in\{1,3\}$ with
$F^{[n-k]}(x)\in[\tfrac{a_k}{4},\tfrac{a_k+1}{4})$ for each $k$;
equivalently, $x$ lies in the cylinder
\[
  C \;:=\; \Bigl[\,z,\; z+4^{-n}\,\Bigr),
  \qquad z\;:=\;\sum_{j=1}^{n} a_j\,4^{-j}\in\onethreefouradicn{n}.
\]
By Lemma~\ref{lem:affine-cylinder}(2), $F^{[n]}$ is affine on $C$ with
slope $4^n$ and vanishes at the left endpoint $z$. Since
$F^{[n]}(x)=0$, we conclude $x=z\in\onethreefouradicn{n}$.

The cardinality statement follows: each digit sequence in $\{1,3\}^n$
yields a distinct zero.
\end{proof}

\subsection{A continuous approximation $F_\varepsilon$}
\label{sec:continuous-F}\label{sec:ccont}

The function $F$ itself cannot be realised by a $\ReLU$ network, since it
is discontinuous at $x=\tfrac14$ and $x=\tfrac34$. We replace it by a
continuous piecewise-affine approximation $F_\varepsilon$ that agrees with
$F$ outside two narrow ``repair intervals'' and linearly interpolates
between the plateau and the ascending piece inside them. Fix a parameter
$\varepsilon\in\bigl(0,\tfrac{1}{16}\bigr)$ and define
\begin{figure}[h]
\centering
\begin{minipage}{0.44\textwidth}
\begin{equation}\label{eq:Feps}
F_\varepsilon(x)=\begin{cases}
1,                & 0\le x\le \tfrac14-\varepsilon,\\[2pt]
\dfrac{1-4x}{4\varepsilon},
                 & \tfrac14-\varepsilon\le x\le \tfrac14,\\[6pt]
4x-1,            & \tfrac14\le x\le \tfrac12,\\[2pt]
1,               & \tfrac12\le x\le \tfrac34-\varepsilon,\\[2pt]
\dfrac{3-4x}{4\varepsilon},
                 & \tfrac34-\varepsilon\le x\le \tfrac34,\\[6pt]
4x-3,            & \tfrac34\le x\le 1.
\end{cases}
\end{equation}
\end{minipage}
\hfill
\begin{minipage}{0.50\textwidth}
\centering
\begin{tikzpicture}[scale=3]
  \draw[->] (-0.02,0) -- (1.08,0) node[right] {$x$};
  \draw[->] (0,-0.02) -- (0,1.08) node[above] {$F_\varepsilon(x)$};
  \foreach \x/\label in {0/0,0.2/{$\tfrac14\!-\!\varepsilon$},0.25/$\tfrac14$,
                         0.5/$\tfrac12$,0.7/{$\tfrac34\!-\!\varepsilon$},
                         0.75/$\tfrac34$,1/1}
    \draw (\x,0.01) -- (\x,-0.01) node[below=2pt,font=\tiny] {\label};
  \foreach \y/\label in {1/1}
    \draw (0.01,\y) -- (-0.01,\y) node[left=2pt] {\label};
  \draw[dashed,gray] (0,1) -- (1,1);
  \draw[thick]
    (0,1) -- (0.2,1)
    (0.2,1) -- (0.25,0)
    (0.25,0) -- (0.5,1)
    (0.5,1) -- (0.7,1)
    (0.7,1) -- (0.75,0)
    (0.75,0) -- (1,1);
\end{tikzpicture}
\end{minipage}
\caption{The continuous piecewise-affine approximation
$F_\varepsilon$ of $F$.}
\label{fig:machine}
\end{figure}

The \emph{repair intervals}
\(
  J_\varepsilon \;:=\; \bigl[\tfrac14-\varepsilon,\tfrac14\bigr]
  \,\cup\,\bigl[\tfrac34-\varepsilon,\tfrac34\bigr]
\)
are the only places where $F_\varepsilon$ and $F$ differ. The map
$F_\varepsilon$ admits the nested form
$F_\varepsilon=\max\bigl(4x-3,\ \min\bigl(\tfrac{3-4x}{4\varepsilon},\ \min\bigl(1,\ \max\bigl(4x-1,\ \min(\tfrac{1-4x}{4\varepsilon},\,1)\bigr)\bigr)\bigr)\bigr)$
(a case inspection over the six pieces), so by Lemma~\ref{lem:CPAM}
it is computable by a $\ReLU$ network in
$\Netnotation{\ReLU}{1}{1}{2}$; composing $n$ copies yields
$F_\varepsilon^{[n]}\in\Netnotation{\ReLU}{1}{1}{2}$ of depth $\mathcal{O}(n)$.

The iterate $F_\varepsilon^{[n]}$ vanishes on every element of
$\onethreefouradicn{n}$ (Lemma~\ref{lem:coincidence}), just as
$F^{[n]}$ does; but continuity forces $F_\varepsilon^{[n]}$ to admit
also \emph{spurious} zeros. Lemma~\ref{lem:localisation} confines
each spurious zero to a small $\varepsilon$-dependent neighbourhood
of a $\{1,3\}$-quadradic of length $s<n$.

\begin{remark}[What is delicate here]
Approximating $F$ by a continuous piecewise-affine map in supremum
norm is immediate and is not the difficulty. Verification depends on
the \emph{zero set of the $n$-fold iterate}, which sup-norm closeness
does not control: each iteration multiplies perturbations by the
Lipschitz constant $4$, and a generic approximation of $F$ acquires
zeros whose orbits encode nothing. The content of this subsection is
that the interpolation \eqref{eq:Feps} confines all spurious zeros of
$F_\varepsilon^{[n]}$ to explicit neighbourhoods of shorter encodings
(Lemma~\ref{lem:localisation}), where the filter of
Lemma~\ref{lem:trunck} repairs them.
\end{remark}

\begin{lemmarep}[Coincidence outside the repair intervals]
\label{lem:coincidence}
Let $x\in[0,1]$. If the $F_\varepsilon$-orbit of $x$ avoids
$J_\varepsilon$ for $n-1$ steps, i.e.\ if
$F_\varepsilon^{[k]}(x)\notin J_\varepsilon$ for every
$k\in\{0,1,\dots,n-1\}$, then
\(
  F_\varepsilon^{[n]}(x) \;=\; F^{[n]}(x).
\)
In particular, $\onethreefouradicn{n}\subseteq(F_\varepsilon^{[n]})^{-1}(0)$.
\end{lemmarep}

\begin{proof}
We prove by induction on $k\le n$ that $F_\varepsilon^{[k]}(x)=F^{[k]}(x)$.
The base case $k=0$ is trivial. Assume the identity at step $k-1$. Then
$F_\varepsilon^{[k-1]}(x)=F^{[k-1]}(x)$, and by hypothesis this point
lies outside $J_\varepsilon$. Since $F_\varepsilon$ and $F$ agree outside
$J_\varepsilon$,
\(
  F_\varepsilon^{[k]}(x) = F_\varepsilon(F_\varepsilon^{[k-1]}(x))
  = F(F^{[k-1]}(x)) = F^{[k]}(x).
\)

For the second assertion, let $x\in\onethreefouradicn{n}$; we show that
$F_\varepsilon^{[n]}(x)=0$. By Lemma~\ref{lem:affine-cylinder},
\[
  F^{[k]}(x)\in\onethreefouradicn{n-k}
  \qquad \text{for every } 0\le k\le n,
\]
and $F^{[n]}(x)=0$ by Lemma~\ref{lem:F-shift}. It therefore suffices to
show that the $F_\varepsilon$-orbit of $x$ coincides with its
$F$-orbit, i.e.\ $F_\varepsilon^{[k]}(x)=F^{[k]}(x)$ for all $k\le n$.

The first part of the lemma does not apply directly: for $k=n-1$ we have
$F^{[n-1]}(x)\in\onethreefouradicn{1}=\{\tfrac14,\tfrac34\}$, which lies
on the boundary of $J_\varepsilon$ rather than outside it. However, at
these two boundary points $F$ and $F_\varepsilon$ still agree: direct
inspection of \eqref{eq:F} and \eqref{eq:Feps} gives
\[
  F(\tfrac14)=F_\varepsilon(\tfrac14)=0,
  \qquad
  F(\tfrac34)=F_\varepsilon(\tfrac34)=0.
\]

We now run the induction of the first part, with the weaker property
that $F$ and $F_\varepsilon$ agree at the point
$F^{[k-1]}(x)$---which holds either because this point lies outside
$J_\varepsilon$, or because it equals $\tfrac14$ or $\tfrac34$. For
$0\le k\le n-1$, the first base-$4$ digit of $F^{[k-1]}(x)\in\onethreefouradicn{n-k+1}$
lies in $\{1,3\}$, placing the point in $[\tfrac14,\tfrac12)\cup[\tfrac34,1)$;
this interval meets $J_\varepsilon=[\tfrac14-\varepsilon,\tfrac14]\cup[\tfrac34-\varepsilon,\tfrac34]$
only at the two boundary points $\tfrac14$ and $\tfrac34$. In either
situation (generic interior point of the admissible interval, or
boundary point), $F$ and $F_\varepsilon$ take the same value there, so
$F_\varepsilon(F^{[k-1]}(x))=F(F^{[k-1]}(x))$, and the induction goes
through. In particular $F_\varepsilon^{[n]}(x)=F^{[n]}(x)=0$.
\end{proof}

\begin{lemmarep}[Localisation of spurious zeros]
\label{lem:localisation}
Let $x\in[0,1]$ with $F_\varepsilon^{[n]}(x)=0$ and
$x\notin\onethreefouradicn{n}$. Then there exist
$s\in\{1,\dots,n\}$ and $z\in\onethreefouradicn{s}$ such that
\(
  |x - z| \;\le\; \varepsilon\cdot 4^{-(s-1)}.
\)
\end{lemmarep}

\begin{proof}
Since $F_\varepsilon^{[n]}(x)=0$ but $x\notin\onethreefouradicn{n}$,
Lemma~\ref{lem:coincidence} implies that the $F_\varepsilon$-orbit of $x$
must enter $J_\varepsilon$ at some step. Let $t\in\{0,\dots,n-1\}$ be the
smallest such step, so that
$F_\varepsilon^{[k]}(x)\notin J_\varepsilon$ for $k<t$ and
$F_\varepsilon^{[t]}(x)\in J_\varepsilon$. By Lemma~\ref{lem:coincidence}
applied at step $t$, $F_\varepsilon^{[t]}(x)=F^{[t]}(x)$, and the orbit
before time $t$ coincides with that of $F$.

The orbit $F^{[0]}(x),\dots,F^{[t-1]}(x)$ all avoid $J_\varepsilon$, so in
particular none of them equals~$1$; an inspection of~\eqref{eq:F} shows
that each $F^{[k]}(x)$ then lies in
$[\tfrac14,\tfrac12)\cup[\tfrac34,1]$, hence has a leading base-$4$ digit
in $\{1,3\}$. Thus the first $t$ digits of $x$ are
$a_1,\dots,a_t\in\{1,3\}$, and by
Lemma~\ref{lem:affine-cylinder}(2), writing
$z_0=\sum_{j=1}^t a_j4^{-j}\in\onethreefouradicn{t}$,
\[
  F^{[t]}(x) \;=\; 4^t(x-z_0) \;\in\; J_\varepsilon.
\]
If $4^t(x-z_0)\in[\tfrac14-\varepsilon,\tfrac14]$, then
$x-z_0\in[4^{-(t+1)}-\varepsilon\cdot 4^{-t},\,4^{-(t+1)}]$, so $x$ is
within $\varepsilon\cdot 4^{-t}$ of the point
$z_0+4^{-(t+1)}\in\onethreefouradicn{t+1}$. The case
$4^t(x-z_0)\in[\tfrac34-\varepsilon,\tfrac34]$ is symmetric and gives the
neighbour $z_0+3\cdot 4^{-(t+1)}\in\onethreefouradicn{t+1}$. Setting
$s:=t+1$ completes the proof.
\end{proof}

Choosing $\varepsilon$ small enough makes the spurious zeros essentially
harmless: they are separated from genuine elements of
$\onethreefouradicn{n}$, and in fact confined to thin one-sided
neighbourhoods of $\{1,3\}$-prefixes. 
Figure~\ref{fig:Feps-iterates} illustrates these spurious zeros for $F_\varepsilon^{[2]}$ and $F_\varepsilon^{[3]}$.
The following quantitative
consequence is what we need in the next subsection.

\input{./fig-feps}

\begin{corollaryrep}\label{cor:eps-choice}
For every $n\ge 1$ and every $\delta>0$, setting
\(
  \varepsilon \;:=\; \delta
\)
guarantees that every zero of $F_\varepsilon^{[n]}$ on $[0,1]$ lies at
distance at most $\delta$ from some element of
$\bigcup_{s=1}^{n}\onethreefouradicn{s}$.
\end{corollaryrep}

\begin{proof}
Let $x\in[0,1]$ be a zero of $F_\varepsilon^{[n]}$. If
$x\in\onethreefouradicn{n}$, then $x$ itself belongs to
$\bigcup_{s=1}^{n}\onethreefouradicn{s}$ and the bound holds with
distance~$0$. Otherwise, Lemma~\ref{lem:localisation} applies: there
exist $s\in\{1,\dots,n\}$ and $z\in\onethreefouradicn{s}$ with
\[
  |x-z| \;\le\; \varepsilon\cdot 4^{-(s-1)} \;\le\; \varepsilon \;=\; \delta,
\]
using $s\ge 1$ for the second inequality.
\end{proof}

\subsection{Lifting the hardness results from $\DDp$ to $[0,1]$}
\label{sec:unrestricted}

We now lift the decision-complexity results of
Section~\ref{sec:variants} from $\onethreefouradicn{n}$ to the whole
of $[0,1]$, at the cost of one additional unit of width. The
fractal gadget of Section~\ref{sec:fractal} produces, from any input
$x\in[0,1]$, an error $p_1 = F_\varepsilon^{[n]}(x)$ witnessing
whether $x$ lies close to the encoding set; combined with the
filter of Lemma~\ref{lem:trunck} and the clause subnetworks of
Section~\ref{sec:reduction}, this yields a width-$4$ network whose
zero test is $\NP$-complete on $[0,1]$. The extra unit of width
carries $p_1$ through the clause stages, so that it can be added
to the clause count at the end.

Only single-quantifier statements ($\NP$, $\coNP$, and gap
inapproximability) transfer cleanly: the counting and majority
results of Theorems~\ref{thm:sharp-p} and~\ref{thm:pp} do not,
since the filter of Lemma~\ref{lem:trunck} is not injective and the
zero set of the lifted network is not in a parsimonious bijection
with the satisfying assignments of $\varphi$.

\noindent \textbf{The lifted network $\NN_\varphi^{[0,1]}$.}
Let $\varphi$ be a 3-CNF formula over $n$ variables with $m$ clauses,
and set
\begin{equation}
  \varepsilon \;:=\; 4^{-(m^2 + m + 2mn + 2n + 2)}. \label{eq:eps-def}
\end{equation}
Define $\NN_\varphi^{[0,1]}\colon[0,1]\to\mathbb{R}_{\ge 0}$ as the
composition of the following four stages:
\begin{enumerate}
\item \emph{Fractal stage.} $x\mapsto (x,\,p_1)$ with
      $p_1 := F_\varepsilon^{[n]}(x)$, computed by the width-$2$
      $\ReLU$ network of Lemma~\ref{lem:CPAM} applied to
      Equation~\eqref{eq:Feps}, with $x$ carried on a third neuron.
      The width is $3$, and the depth is $\mathcal{O}(n)$.
\item \emph{Filter stage.} $(x,\,p_1)\mapsto (\tilde x,\,p_1)$ via the
      truncation network $\NNT{n}$ of Lemma~\ref{lem:trunck} applied
      to $x$, with $p_1$ carried untouched on a dedicated register. Width   $4$,   depth   $\mathcal{O}(n)$.
\item \emph{Clause stages.} $(\tilde x,\,0,\,p_1)\mapsto
      (\tilde x,\,q_1,\,p_1)\mapsto\cdots\mapsto
      (\tilde x,\,q_{m+1},\,p_1)$, where each transition applies a
      clause subnetwork $\NN_{\varphi,\ell}$ of
      Section~\ref{sec:reduction} to $(\tilde x,q_{\ell-1})$,
      producing $(\tilde x,q_\ell)$, while $p_1$ is passed through on
      a dedicated register. The width is $4$, and the  depth is $\mathcal{O}(mn)$.
\item \emph{Output combination.} Return
      \(
        \NN_\varphi^{[0,1]}(x) \;:=\;
        \NNCOUNT{m}(q_{m+1},\,0) \,+\, p_1,
      \)
      where $\NNCOUNT{m}$ is the counting gadget of
      Lemma~\ref{lem:counting}. The width is $4$, and the  depth is $\mathcal{O}(m)$.
\end{enumerate}
The network $\NN_\varphi^{[0,1]}$ lies in
$\Netnotation{\ReLU}{1}{1}{4}$, has depth $\mathcal{O}(mn)$, and all
rational weights have encoding size polynomial in $n+m$ (in particular,
$\varepsilon$ has encoding size $\mathcal{O}(mn)$ bits). The
construction is computable polynomial-time from $\varphi$.

\begin{mytheoremrep}[Correctness of $\NN_\varphi^{[0,1]}$]\label{thm:lifted}
For every $x\in[0,1]$, $\NN_\varphi^{[0,1]}(x)\ge 0$. Moreover:
\begin{enumerate}
\item If $\varphi$ is satisfiable, then for every satisfying
      assignment $\alpha\in\{0,1\}^n$,
      $\NN_\varphi^{[0,1]}(x(\alpha)) = 0$.
\item If $\varphi$ is unsatisfiable, then
      $\NN_\varphi^{[0,1]}(x)\ge\tfrac12$ for every $x\in[0,1]$.
\end{enumerate}
In particular,
\(
  \exists\,x\in[0,1]\colon\NN_\varphi^{[0,1]}(x) = 0
  \;\Longleftrightarrow\;
  \varphi \text{ is satisfiable.}
\)
\end{mytheoremrep}

\begin{proofsketch}
Part~(1) is immediate: on $x = x(\alpha)$, the fractal error
$p_1=F_\varepsilon^{[n]}(x)$ equals $0$
(Lemma~\ref{lem:coincidence}) and the clause output is $0$
(Claim~\ref{claim:assignment}), so
$\NN_\varphi^{[0,1]}(x)=0$. Part~(2) splits into three cases on
the geometry of~$x$: \emph{(A)}~if
$x\in\onethreefouradicn{n}$, unsatisfiability forces at least one
$c_\ell=3$; \emph{(B)}~if $p_1\ge\tfrac12$, the output dominates;
\emph{(C)}~otherwise, Lemma~\ref{lem:localisation} places~$x$
within $\varepsilon$ of some $z\in\onethreefouradicn{t+1}$, and
the Lipschitz bounds $L_{\mathrm{count}}\le 4^{m^2+1}$,
$L_{\mathrm{clauses}}\le 4^{2mn}$, $L_{\mathrm{trunck}}\le  {2n}$,
  combined with the choice of $\varepsilon$ in
  Equation~\eqref{eq:eps-def}, which gives $\Lambda\varepsilon\le\tfrac1{16}$ for
  $\Lambda:=L_{\mathrm{count}}L_{\mathrm{clauses}}L_{\mathrm{trunck}}$, keep the output within
\end{proofsketch}

\begin{proof}
Non-negativity is immediate: $\NNCOUNT{m}(q_{m+1},0)\ge 0$ is a
non-negative integer by Lemma~\ref{lem:counting}, and $p_1\ge 0$ since
$F_\varepsilon^{[n]}$ takes values in $[0,1]$.

\smallskip
\paragraph*{Lipschitz constants.}
Each of the three main components is Lipschitz over $[0,1]$ (or
$[0,1]^2$ as appropriate):
\begin{itemize}
\item The filter $\NNT{n}$ of Lemma~\ref{lem:trunck}, viewed as a map
      $x\mapsto\tilde x$, is $L_{\mathrm{trunck}}$-Lipschitz with
      $L_{\mathrm{trunck}}\le {2n}$
      (Lemma~\ref{lem:trunck-lipschitz}: the weight $4^{-i}$ of the
      digit written at step $i$ exactly offsets the drift $4^{i-1}$
      of the register it is read from).
\item The clause network $\NN_\varphi := \NN_{\varphi,m}\circ\cdots\circ
      \NN_{\varphi,1}$, viewed as the map $(\tilde x,q_0)\mapsto
      q_{m+1}$ via the clause stages, is
      $L_{\mathrm{clauses}}$-Lipschitz with
      $L_{\mathrm{clauses}}\le 4^{2mn}$ (each of the $m$ clauses scans
      $n$ digits with per-step Lipschitz $\le 16$).
\item The counting gadget $\NNCOUNT{m}(\cdot,0)$ is
      $L_{\mathrm{count}}$-Lipschitz with $L_{\mathrm{count}}\le 4^{m^2+1}$
      (each of the $m$ iterations scans a digit with per-step Lipschitz
      $\le 4^{m+1}$ in the transition zones of the shift function
      $s(\cdot)$ of Lemma~\ref{lem:counting}).

Setting $\Lambda := L_{\mathrm{count}}\cdot L_{\mathrm{clauses}}\cdot
L_{\mathrm{trunck}}$, our choice of $\varepsilon$
satisfies
\begin{equation}\label{eq:epsilon-bound}
  \Lambda\cdot\varepsilon
  \;=\; 2n\cdot 4^{\,m^2+1+2mn-(m^2+m+2mn+2n+2)}
  \;=\; 2n\cdot 4^{-(m+2n+1)}
  \;\le\;
  4^{-2}
  \;=\;
  \tfrac{1}{16},
\end{equation}
using $2n\le 4^{\,m+2n-1}$ for all $m,n\ge 1$.
\end{itemize}

\smallskip
\paragraph*{Part (1): $\varphi$ satisfiable.}
Let $\alpha$ be a satisfying assignment and let $x = x(\alpha)\in
\onethreefouradicn{n}$. By Lemma~\ref{lem:coincidence} and
Proposition~\ref{lem:F-shift}, $p_1 = F_\varepsilon^{[n]}(x) =
F^{[n]}(x) = 0$. By Lemma~\ref{lem:trunck}(1), $\tilde x = x$. By
Claim~\ref{claim:assignment} applied with initial flag $q_0 = 0$, the
clause stages produce $q_{m+1} = \sum_\ell c_\ell\,4^{-(m+1-\ell)}$
with every $c_\ell = 0$ (since $\alpha$ satisfies every clause). Hence
$q_{m+1} = 0$, and $\NNCOUNT{m}(0,0) = 0$ by
Lemma~\ref{lem:counting}. Therefore $\NN_\varphi^{[0,1]}(x) = 0$.

\smallskip
\paragraph*{Part (2): $\varphi$ unsatisfiable.}
Fix $x\in[0,1]$. We show $\NN_\varphi^{[0,1]}(x)\ge\tfrac12$ by
distinguishing three cases.

\smallskip
\emph{Case A: $x\in\onethreefouradicn{n}$.}
Then $p_1 = 0$ and $\tilde x = x$ by Lemma~\ref{lem:trunck}(1). Let
$\alpha$ be the assignment associated to $x$. Since $\varphi$ is
unsatisfiable, $\alpha$ fails at least one clause, so by
Claim~\ref{claim:assignment} there exists $\ell$ with $c_\ell = 3$.
Hence $\NNCOUNT{m}(q_{m+1},0)\ge 1$ by Lemma~\ref{lem:counting}, and
$\NN_\varphi^{[0,1]}(x)\ge 1$.

\smallskip
\emph{Case B: $p_1\ge\tfrac12$.}
Then $\NN_\varphi^{[0,1]}(x)\ge p_1\ge\tfrac12$.

\smallskip
\emph{Case C: $x\notin\onethreefouradicn{n}$ and $p_1 < \tfrac12$.}
We distinguish two sub-cases according to whether the
$F_\varepsilon$-orbit of $x$ enters the repair region $J_\varepsilon$.

\smallskip
\emph{Sub-case C.1: The orbit avoids $J_\varepsilon$ at every step
$0,\dots,n-1$.}
By Lemma~\ref{lem:coincidence}, $F_\varepsilon^{[n]}(x) = F^{[n]}(x)$,
and this value equals $p_1 < \tfrac12$. We claim that the first $n$
base-$4$ digits of $x$ lie in $\{1,3\}$.

Indeed, if some orbit point $F^{[k]}(x)$ for
$k\in\{0,\dots,n-1\}$ lies in the plateau region
$[0,\tfrac14)\cup[\tfrac12,\tfrac34)$ where $F\equiv 1$, then
$F^{[k+1]}(x) = 1$; and since $F(1) = 1$, every subsequent iterate
equals $1$, giving $F^{[n]}(x) = 1$ and contradicting
$p_1 < \tfrac12$. Hence each orbit point
$F^{[k]}(x)$ lies in $[\tfrac14,\tfrac12)\cup[\tfrac34,1)$, so has
leading base-$4$ digit in $\{1,3\}$ (the boundary points
$\tfrac14,\tfrac34$ are excluded since the orbit avoids
$J_\varepsilon$). Thus the first $n$ digits of $x$ are
$a_1,\dots,a_n\in\{1,3\}$.

By Lemma~\ref{lem:trunck}(1),
$\tilde x = \sum_{j=1}^n a_j\,4^{-j}\in\onethreefouradicn{n}$. Let
$\alpha$ be the associated assignment. Since $\varphi$ is
unsatisfiable, $\alpha$ fails some clause, so by
Claim~\ref{claim:assignment}, some $c_\ell = 3$. Hence
$\NNCOUNT{m}(q_{m+1},0)\ge 1$ by Lemma~\ref{lem:counting}, and
$\NN_\varphi^{[0,1]}(x)\ge 1 + p_1 \ge 1$.

\smallskip
\emph{Sub-case C.2: The orbit enters $J_\varepsilon$ at some step.}
Let $t\in\{0,\dots,n-1\}$ be the smallest such step. By the argument
of Lemma~\ref{lem:localisation}, the first $t$ digits of $x$ are in
$\{1,3\}$, and there exists $z\in\onethreefouradicn{t+1}$ with
$|x - z|\le\varepsilon\cdot 4^{-t}\le\varepsilon$.

Applying the filter to $z$ yields
$\tilde z := \NNT{n}(z,0) \in \onethreefouradicn{n}$ by
Lemma~\ref{lem:trunck}(2) (since $z$ has digits in $\{1,3\}$ at
positions $\le t+1$ and $0$ at positions $> t+1$, all in $\{0,1,3\}$).
Let $\alpha_z$ be the assignment associated to $\tilde z$ and let
$q_{m+1}^*$ denote the clause output at $\tilde z$. Since $\varphi$
is unsatisfiable, $\alpha_z$ fails some clause, so by
Claim~\ref{claim:assignment} some $c_\ell^* = 3$, and
$\NNCOUNT{m}(q_{m+1}^*,0)\ge 1$ by Lemma~\ref{lem:counting}.

By Lipschitz continuity of the three components,
\begin{align*}
  |\tilde x - \tilde z|
    &\;\le\; L_{\mathrm{trunck}}\cdot|x - z|
     \;\le\; L_{\mathrm{trunck}}\cdot\varepsilon,\\
  |q_{m+1} - q_{m+1}^*|
    &\;\le\; L_{\mathrm{clauses}}\cdot|\tilde x - \tilde z|
     \;\le\; L_{\mathrm{clauses}}\cdot L_{\mathrm{trunck}}\cdot\varepsilon,\\
  |\NNCOUNT{m}(q_{m+1},0) - \NNCOUNT{m}(q_{m+1}^*,0)|
    &\;\le\; L_{\mathrm{count}}\cdot|q_{m+1} - q_{m+1}^*|
     \;\le\; \Lambda\cdot\varepsilon
     \;\le\; \tfrac{1}{16},
\end{align*}
where the last inequality uses~\eqref{eq:epsilon-bound}. Hence
\[
  \NNCOUNT{m}(q_{m+1},0)
  \;\ge\; \NNCOUNT{m}(q_{m+1}^*,0) - \tfrac{1}{16}
  \;\ge\; 1 - \tfrac{1}{16}
  \;=\; \tfrac{15}{16},
\]
and therefore
$\NN_\varphi^{[0,1]}(x) \ge \tfrac{15}{16} + p_1 > \tfrac12$.

\smallskip
Combining Cases A, B, and C, we conclude
$\NN_\varphi^{[0,1]}(x)\ge\tfrac12$ for every $x\in[0,1]$ whenever
$\varphi$ is unsatisfiable.

\begin{lemma}[Lipschitz constant of the truncation filter]
\label{lem:trunck-lipschitz}
The network $\NNT{n}\colon[0,1]\times\mathbb{R}\to\mathbb{R}^2$ of
Lemma~\ref{lem:trunck} is globally $L_{\mathrm{trunck}}$-Lipschitz in
its first argument $x$, with
\[
  L_{\mathrm{trunck}} \;\le\;  2n.
\]
\end{lemma}
 \begin{nestedproof}
     Recall from Lemma~\ref{lem:trunck} that, starting from
$(x_0, y_0) = (x, 0)$, the loop iterates
\[
  y_i \;=\; \ReLU\bigl(y_{i-1} + \lead(x_{i-1})\,4^{-i}\bigr),
  \qquad
  x_i \;=\; \shift(x_{i-1}),
  \qquad 1 \le i \le n,
\]
with $\shift(x) = \ReLU(4x - \lead(x))$, and outputs $y_n$. All
constants below are maximum absolute slopes of continuous
piecewise-affine maps, hence Lipschitz constants.

\smallskip

The function $\lead$ has maximum absolute slope $8$
(Lemma~\ref{lem:lead}). The map $x \mapsto 4x - \lead(x)$ is
continuous piecewise-affine with slopes $4$, $-4$, $4$ on the pieces
$[0,\tfrac12]$, $[\tfrac12,\tfrac34]$, $[\tfrac34,1]$, so it has
maximum absolute slope $4$; composing with the $1$-Lipschitz $\ReLU$,
the map $\shift$ is $4$-Lipschitz. The accumulator update
$(y, x) \mapsto \ReLU\bigl(y + \lead(x)\,4^{-i}\bigr)$ is
$1$-Lipschitz in $y$ and $8 \cdot 4^{-i}$-Lipschitz in $x$.

\smallskip

Fix two inputs $x_0, x_0' \in [0,1]$ and write $x_i, x_i'$ and
$y_i, y_i'$ for the corresponding register values (with
$y_0 = y_0' = 0$). By the $4$-Lipschitz bound on $\shift$ and
induction,
\[
  |x_i - x_i'| \;\le\; 4^{\,i}\,|x_0 - x_0'|.
\]
The accumulator recurrence then gives
\[
  |y_i - y_i'|
  \;\le\; |y_{i-1} - y_{i-1}'| + 8 \cdot 4^{-i}\,|x_{i-1} - x_{i-1}'|
  \;\le\; |y_{i-1} - y_{i-1}'| + 8 \cdot 4^{-i} \cdot 4^{\,i-1}\,|x_0 - x_0'|,
\]
and $8 \cdot 4^{-i} \cdot 4^{\,i-1} = 2$ for every $i$: the weight
$4^{-i}$ of the digit written at step $i$ exactly offsets the
accumulated drift $4^{\,i-1}$ of the register it is read from. Summing
over $i = 1, \dots, n$ with $y_0 = y_0'$,
\[
  |y_n - y_n'| \;\le\; 2n\,|x_0 - x_0'|.
\]

 \end{nestedproof}

\begin{lemma}[Lipschitz constant of the clause cascade]
\label{lem:clauses-lipschitz}
The network $\NN_\varphi\colon[0,1]\times[0,1]\to[0,1]$ obtained by
composing the $m$ clause subnetworks of
Section~\ref{sec:reduction} is globally
$L_{\mathrm{clauses}}$-Lipschitz, with
\[
  L_{\mathrm{clauses}} \;\le\; 4^{2mn}.
\]
\end{lemma}
 
\begin{nestedproof}
Each clause subnetwork $\NN_{\varphi,\ell}$ consists of an affine
flag-initialisation step followed by $n$ inner blocks
$\NN_{\varphi,\ell,1},\dots,\NN_{\varphi,\ell,n}$, each of which
performs (i)~a digit extraction via $\lead$, (ii)~a flag update via
$\test$, and (iii)~a shift followed by digit reinsertion. We bound
each step.

\begin{itemize}
\item  The function $\lead$ has
      maximum absolute slope $8$ and range $[1,3]$
      (Lemma~\ref{lem:lead}). Hence
      $b(x) := \ReLU\bigl(\tfrac{3-\lead(x)}{2}\bigr)
             = \tfrac{3-\lead(x)}{2}\in[0,1]$
      (the $\ReLU$ is inactive) is $4$-Lipschitz. The literal gate
      $\ellit_{\ell,j}$ is one of $b\mapsto\ReLU(b)$,
      $b\mapsto\ReLU(1-b)$, or the constant $0$, hence $1$-Lipschitz;
      therefore $e_j(x) := \ellit_{\ell,j}(b(x))$ is $4$-Lipschitz in
      $x$, with values in $[0,1]$.
\item  The function $g$ of Lemma~\ref{lem:test}
      has maximum absolute slope $16$, so
      $\test(e,p) = g\bigl(\tfrac{e}{2}+\tfrac{p}{16}\bigr)$ is
      $8$-Lipschitz in $e$ and $1$-Lipschitz in $p$. Composed with
      $e_j$, the flag update of one inner block is
      $8\cdot 4 = 32$-Lipschitz in the block input $x$ and
      $1$-Lipschitz in the incoming flag.
\item The register update of one inner block is
      $h(x) := \shift(x) + 4^{-n}\lead(x)$, applied to the block
      input $x$, since the register is not modified between the
      extraction and the rotation. With
      $\shift(x)=\ReLU(4x-\lead(x))$, the map $h$ is continuous
      piecewise-affine:
      \[
        h(x)=
        \begin{cases}
          4^{-n} & 0\le x\le\tfrac14,\\
          4x-1+4^{-n} & \tfrac14\le x\le\tfrac12,\\
          3-4x+4^{-n}(8x-3) & \tfrac12\le x\le\tfrac34,\\
          4x-3+3\cdot 4^{-n} & \tfrac34\le x\le 1,
        \end{cases}
      \]
      with slopes $0$, $4$, $-4+8\cdot 4^{-n}$, $4$. Since $n\ge 1$,
      we have $8\cdot 4^{-n}\le 2$, so
      $\lvert -4+8\cdot 4^{-n}\rvert\le 4$: every slope has absolute
      value at most $4$, and $h$ is $4$-Lipschitz. The same bound
      holds for the variant of $\shift$ without the outer $\ReLU$,
      since $\ReLU$ is $1$-Lipschitz.
\item   The map
      $p\mapsto\min\bigl(\tfrac34+\tfrac{p}{4},\,1\bigr)$ is
      $\tfrac14$-Lipschitz in $p$ and does not read $x$.
\end{itemize}

Setting $L_{\mathrm{block}} := 16$ as the per-inner-block Lipschitz
constant on the joint state $(x,p)$, each clause subnetwork
$\NN_{\varphi,\ell}$ is $L_{\mathrm{block}}^n = 16^n = 4^{2n}$-Lipschitz.
Composing the $m$ clause subnetworks yields
\[
  L_{\mathrm{clauses}} \;\le\; (4^{2n})^m \;=\; 4^{2mn}.
\]
\end{nestedproof}

\begin{lemma}[Lipschitz constant of the counter]
\label{lem:counting-lipschitz}
The function $\NNCOUNT{n}\colon[0,1]\times\mathbb{R}\to\mathbb{R}$
constructed in Lemma~\ref{lem:counting} is globally
$L_{\mathrm{count}}$-Lipschitz, with
\[
L_{\mathrm{count}} \;\le\;  4^{\,n^2+1}.
\]
More precisely, $\NNCOUNT{n}$ is $1$-Lipschitz in its second argument
$p$, and $4^{n^2+1}$-Lipschitz in its first argument $x$.
\end{lemma}
 
\begin{nestedproof}
Recall from Lemma~\ref{lem:counting} that $\NNCOUNT{n}$ iterates, from
$(x_0, c_0) := (x, p)$,
\[
  c_i \;=\; c_{i-1} + \tfrac12\bigl(\lead(x_{i-1}) - 1\bigr),
  \qquad
  x_i \;=\; \ReLU\bigl(4x_{i-1} - s(x_{i-1})\bigr),
  \qquad 1 \le i \le n,
\]
and outputs $c_n$. For $n = 0$ the output equals $p$ and all claims are
trivial, so assume $n \ge 1$.

The $1$-Lipschitz bound in $p$ is immediate: the second argument is
copied to $c$ at step~1 and only modified additively by quantities
depending on $x$ alone.

For the bound in $x$, we control the Lipschitz constant of each
operation in the loop.
\begin{itemize}
\item The function $\lead$ has maximum absolute
      slope $8$ (Lemma~\ref{lem:lead}), so the increment map
      $x \mapsto \tfrac12(\lead(x) - 1)$ is $4$-Lipschitz.

\item The shift function $s\colon[0,1]\to\mathbb{R}$ defined in the
      proof of Lemma~\ref{lem:counting} has maximal slope
      $\max(1/\eta,\,2/(\tfrac14+\eta))$ with $\eta=4^{-(n+1)}$,
      hence $4^{n+1}$-Lipschitz.
\item The function  $x \longmapsto \ReLU\bigl(4x-s(x)\bigr)$ 
is $4^{n+1}$-Lipschitz and, in particular,$4+4^{n+1}$-Lipschitz.
Indeed, define $u(x):=4x-s(x)$, $x\in[0,1]$. We compute the slope for each  piece,
\begin{itemize}
\item on $[0,\,\tfrac14-\eta]$, slope $4$;
\item on $[\tfrac14-\eta,\,\tfrac14]$, slope $4-1/\eta = 4- 4^{n+1}$;
\item on $[\tfrac14,\,\tfrac12-\eta]$, slope $4$;
\item on $[\tfrac12-\eta,\,\tfrac34]$, slope $4- 8/(\tfrac14+\eta)$
\end{itemize}
Hence, the Lipschitz constant of $u$ is $4^{n+1}-4$.
\end{itemize}
Since $\ReLU$ is
      $1$-Lipschitz, the update $x \mapsto \ReLU(u(x))$ is
      $L_x$-Lipschitz with $L_x := 4^{n+1}$. 

Fix $p$ and two inputs $x_0, x_0' \in [0,1]$; write $x_i, x_i'$ and
$c_i, c_i'$ for the corresponding register values. By induction on $i$,
using the $1$-Lipschitzness of $\ReLU$ and the bound on $u$,
\[
  |x_i - x_i'| \;\le\; L_x^{\,i}\,|x_0 - x_0'|.
\]
The counter is updated by
$c_i = c_{i-1} + \tfrac12(\lead(x_{i-1})-1)$, and this update is
$4$-Lipschitz in $x_{i-1}$, so
\[
  |c_i - c_i'| \;\le\; |c_{i-1}-c_{i-1}'| + 4\,|x_{i-1}-x_{i-1}'|
                \;\le\; |c_{i-1}-c_{i-1}'| + 4\,L_x^{i-1}\,|x_0-x_0'|.
\]
Summing over $i = 1, \dots, n$ with $c_0 = c_0' = p$, and using
$L_x \ge 1$ to bound each term of the geometric sum by its last term,
\[
  |c_n - c_n'|
  \;\le\; 4\,|x_0 - x_0'| \sum_{i=0}^{n-1} L_x^{\,i}
  \;\le\;4\frac{L^{n}_{x }-1}{L_{x }-1} \,|x_0 - x_0'|
\]
Since $L_x = 4^{n+1} \ge 16$, we have $L_x - 1 \ge \tfrac{15}{16}\,L_x$,
hence
\[
  4\,\frac{L_x^{\,n} - 1}{L_x - 1}
  \;\le\; \tfrac{16}{15}\cdot 4\,L_x^{\,n-1}
  \;=\; \tfrac{16}{15}\cdot 4^{\,(n+1)(n-1)+1}
  \;=\; \tfrac{16}{15}\cdot 4^{\,n^2}
  \;\le\; 4^{\,n^2+1}
\]
which proves the claimed bound. 
\end{nestedproof}

\end{proof}

\noindent\textbf{Consequences.}
The correctness theorem yields the expected hardness results for
verification on $[0,1]$ at width $4$.

\begin{corollaryrep}[\textsc{Verification} and \textsc{Universal-Gap} on {$[0,1]$}]
\label{cor:np-conp-lifted}
The following problems are complete for their respective classes, for
width-$4$ $\ReLU$ networks on $[0,1]$:
\begin{itemize}
\item ($\NP$) given $\NN\in\Netnotation{\ReLU}{1}{1}{4}$, decide
      whether $\exists x\in[0,1]\colon\NN(x) = 0$;
\item ($\coNP$) given the same input, decide whether
      $\NN(x)\ge\tfrac12$ for every $x\in[0,1]$.
\end{itemize}
\end{corollaryrep}

\begin{proof}
Membership is inherited from the unrestricted verification problem
of~\cite{wurm1}. For hardness, map
$\varphi\mapsto\NN_\varphi^{[0,1]}$: Theorem~\ref{thm:lifted} gives
the $\NP$ reduction from \textsc{3-Sat}, and by the same theorem
$\NN_\varphi^{[0,1]}(x)\ge\tfrac12$ for every $x\in[0,1]$ iff
$\varphi$ is unsatisfiable, which reduces \textsc{3-Unsat}.
\end{proof}

\begin{corollaryrep}[Uniform distance on {$[0,1]$}]
\label{cor:uniform-lifted}
Given a width-$4$ $\ReLU$ network $\NN\colon[0,1]\to\mathbb{R}$ and a
rational $\tau>0$, deciding whether
$\displaystyle\sup_{x\in[0,1]}|\NN(x)|\le\tau$ is $\coNP$-complete,
even for $\tau=\tfrac18$.
\end{corollaryrep}

\begin{proof}
Consider the auxiliary network
$\widetilde\NN_\varphi(x) := \ReLU\bigl(\tfrac14 - \NN_\varphi^{[0,1]}(x)\bigr)$,
which remains of width $4$. By Theorem~\ref{thm:lifted},
$\NN_\varphi^{[0,1]}(x)\in\{0\}\cup[\tfrac12,\infty)$, so
$\widetilde\NN_\varphi(x)\in\{0,\tfrac14\}$: it equals $\tfrac14$
exactly on the zeros of $\NN_\varphi^{[0,1]}$, and equals $0$
elsewhere. Hence $\sup_{x\in[0,1]}\widetilde\NN_\varphi(x)$ equals
$\tfrac14$ if $\NN_\varphi^{[0,1]}$ has a zero on $[0,1]$ and $0$
otherwise; it is therefore $\le\tfrac18$ iff (by
Theorem~\ref{thm:lifted}) $\varphi$ is unsatisfiable.
\end{proof}
\begin{corollaryrep}[Gap inapproximability on {$[0,1]$}]
\label{cor:inapprox-lifted}
For every constant $\eta>0$, the following promise problem is
$\NP$-hard: given a width-$4$ $\ReLU$ network
$\NN\colon[0,1]\to\mathbb{R}_{\ge 0}$ and integers $n,m$, distinguish
\(
  \min_{x\in[0,1]}\NN(x) = 0\)
  from 
\(
  \min_{x\in[0,1]}\NN(x)\ge\bigl(\tfrac{1}{8}-\eta\bigr)\,m.
\)
\end{corollaryrep}

\begin{proof}
H{\aa}stad's theorem~\cite{hastad2001some} gives the analogous gap on
\textsc{Max-3Sat}: satisfiable formulas versus those with at most
$(\tfrac78+\eta)m$ clauses simultaneously satisfiable. We reduce this
gap version of \textsc{Max-3Sat} to the claimed promise problem.

For a 3-CNF formula $\varphi$ with $n$ variables and $m$ clauses, define
the auxiliary network
\[
  \widehat{\NN}_\varphi(x)
  \;:=\;
  \NNCOUNT{m}(q_{m+1},\,0) \;+\; m\cdot p_1,
\]
where $q_{m+1}$ and $p_1$ are the clause output and fractal error
produced by the same fractal, filter, and clause stages as in
$\NN_\varphi^{[0,1]}$ (see the construction preceding
Theorem~\ref{thm:lifted}). The only difference with
$\NN_\varphi^{[0,1]}$ is that the fractal contribution $p_1$ is
scaled by the constant $m$ in the final output combination. The
scaling by $m$ is a rational affine operation that does not increase
the width, so $\widehat{\NN}_\varphi\in\Netnotation{\ReLU}{1}{1}{4}$ and
is computable from $\varphi$ in polynomial time. The entire analysis of
Theorem~\ref{thm:lifted} transfers verbatim with $p_1$ replaced by
$m\cdot p_1$ in the output; in particular, $\widehat{\NN}_\varphi(x)=0$
iff $\varphi$ is satisfiable and $x=x(\alpha)$ for some satisfying
$\alpha$.

\smallskip
\emph{YES case.} If $\varphi$ is satisfiable, let $\alpha$ be a
satisfying assignment and $x=x(\alpha)$. By the same argument as in
Theorem~\ref{thm:lifted}(1), $p_1=0$ and $q_{m+1}=0$, so
$\widehat{\NN}_\varphi(x)=0$. Hence $\min_{x\in[0,1]}\widehat{\NN}_\varphi(x)=0$.

\smallskip
\emph{NO case.} Suppose that no assignment satisfies more than
$(\tfrac78+\eta)m$ clauses of $\varphi$, so every assignment leaves at
least $(\tfrac18-\eta)m$ clauses unsatisfied. We show
$\widehat{\NN}_\varphi(x)\ge(\tfrac18-\eta')m$ for every $x\in[0,1]$,
for some $\eta'$ slightly larger than $\eta$. We distinguish the same
three cases as in the proof of Theorem~\ref{thm:lifted}(2).

\emph{Case A} ($x\in\onethreefouradicn{n}$). Then $p_1=0$ and
$\tilde x=x$, so $\tilde x$ corresponds to some assignment $\alpha$.
By Claim~\ref{claim:assignment} and Lemma~\ref{lem:counting},
$\NNCOUNT{m}(q_{m+1},0)=|\{\ell:c_\ell=3\}|$, the number of clauses
of $\varphi$ unsatisfied by $\alpha$. By the \textsc{Max-3Sat} gap,
this is $\ge(\tfrac18-\eta)m$, so
$\widehat{\NN}_\varphi(x)\ge(\tfrac18-\eta)m$.

\emph{Case B} ($p_1\ge\tfrac12$). Then
$\widehat{\NN}_\varphi(x)\ge m\cdot p_1\ge m/2\ge(\tfrac18-\eta)m$
for every $\eta\le\tfrac38$ (and hence in particular for every
$\eta\in(0,\tfrac18)$, which is the regime of interest for the gap).

\emph{Case C} ($x\notin\onethreefouradicn{n}$ and $p_1<\tfrac12$).
\emph{Sub-case C.1} (orbit avoids $J_\varepsilon$): by the same
argument as in Theorem~\ref{thm:lifted}(2), $\tilde x\in\onethreefouradicn{n}$
corresponds exactly to some assignment $\alpha$, and $\alpha$ leaves
at least $(\tfrac18-\eta)m$ clauses unsatisfied by the gap hypothesis.
Hence $\NNCOUNT{m}(q_{m+1},0)\ge(\tfrac18-\eta)m$, and
$\widehat{\NN}_\varphi(x)\ge(\tfrac18-\eta)m$.

\emph{Sub-case C.2} (orbit enters $J_\varepsilon$): there exists
$z\in\onethreefouradicn{t+1}$ within $\varepsilon$ of $x$, and the
assignment $\alpha_z$ associated to $\tilde z:=\NNT{n}(z,0)\in
\onethreefouradicn{n}$ leaves at least $(\tfrac18-\eta)m$ clauses
unsatisfied, so $\NNCOUNT{m}(q_{m+1}^*,0)\ge(\tfrac18-\eta)m$. The
Lipschitz computation of Theorem~\ref{thm:lifted}(2) gives
\[
  |\NNCOUNT{m}(q_{m+1},0) - \NNCOUNT{m}(q_{m+1}^*,0)|
  \;\le\; \tfrac{1}{16},
\]
hence $\NNCOUNT{m}(q_{m+1},0)\ge(\tfrac18-\eta)m - \tfrac{1}{16}$.
For any $\eta'>\eta$ with $\eta'-\eta\ge\tfrac{1}{16m}$, this is
$\ge(\tfrac18-\eta')m$, and since $m\ge 1$ and $\eta'$ can be chosen
arbitrarily close to $\eta$, we obtain
$\widehat{\NN}_\varphi(x)\ge(\tfrac18-\eta')m$.

\smallskip
Taking the minimum across Cases A, B, C yields
$\min_{x\in[0,1]}\widehat{\NN}_\varphi(x)\ge(\tfrac18-\eta')m$, as
required.
\end{proof}

\begin{remark}[Width-$3$ lifting: an open problem]
Whether the lifting above can be carried out at width $3$ remains
open. The difficulty is genuine: at width $3$, there is no room to
carry the fractal error $p_1$ through the clause stages, and using
$p_1$ as the initial flag of the first clause stage (which preserves
width $3$) breaks correctness on inputs whose
$F_\varepsilon$-orbit enters the repair region $J_\varepsilon$ from
the side yielding digit-$2$ prefixes. 
\end{remark}

\subsection{The other end of the scale: width $1$}
\label{sec:width-one}

\begin{propositionrep}\label{prop:width-one}
A width-$1$ $\ReLU$ network $\NN\colon\mathbb{R}\to\mathbb{R}$ of
depth $L$ computes a \emph{monotone} continuous piecewise-affine
function with at most $2L+2$ pieces, computable in polynomial time.
Consequently, deciding $\exists x\in[0,1]\colon\NN(x)=0$ and
$\exists x\in\onethreefouradicn{n}\colon\NN(x)=0$ are in $\P$ for
width-$1$ networks.
\end{propositionrep}
\begin{proof}
Each hidden layer applies $h\mapsto\ReLU(ah+b)$, a monotone map of a
scalar; compositions of monotone maps are monotone, and the outer
affine maps preserve this. Applying $\ReLU$ to a monotone
piecewise-affine function adds at most two breakpoints (the endpoints
of the interval where the affine image is non-positive), all rational
of polynomial bit-size. The zero set of the resulting monotone
function on $[0,1]$ is a rational interval $[u,v]$;
$[u,v]\cap[0,1]\neq\emptyset$ is immediate, and
$[u,v]\cap\onethreefouradicn{n}\neq\emptyset$ is a greedy
digit-by-digit descent in base $4$.
\end{proof}

The tractability persists in the dynamical reading, where the network
is iterated a number of times given in unary:

\begin{mytheoremrep}
     \label{thdimone}
\textsc{Verification} for iterated one-dimensional $\ReLU$ width-$1$
networks with unary time bound is decidable in polynomial time.
\end{mytheoremrep}

\begin{proof}
We consider $\ReLU$ neural networks
\(\NN \colon [0,1] \to [0,1]\)
with rational weights and biases. Such a network induces a continuous function
\(f = \NN\) from \([0,1]\) to \([0,1]\).
We study the following bounded-time reachability problem.

\medskip
\paragraph*{1D Reachability (unary time).}
\begin{itemize}
\item \emph{Input:} a $\ReLU$ network \(\NN:[0,1]\to[0,1]\); two intervals
\(R_0,R\subseteq[0,1]\); a time bound \(T\in\mathbb{N}\) encoded in unary.
\item \emph{Question:} does there exist \(x\in R_0\) and \(t\le T\) such that
\(\NN^{[t]}(x)\in R\)?
\end{itemize}


It is well known that any $\ReLU$ network computes a continuous piecewise-affine function.
In dimension one, this structure is particularly simple.

\begin{lemma}
Let \(\NN:[0,1]\to[0,1]\) be a $\ReLU$ network of width 1.
Then there exists a finite partition
\(
0 = p_0 < p_1 < \cdots < p_m = 1
\)
such that on each interval \([p_i,p_{i+1}]\), the function \(f=\NN\) is affine:
\(
f(x) = a_i x + b_i\)  \text{for } \(x\in[p_i,p_{i+1}].
\)
Moreover, this representation can be computed in time polynomial in the size of \(\NN\).
\end{lemma}

\begin{nestedproof}
The proof proceeds layer by layer.
Affine layers preserve piecewise-affineness and only merge existing breakpoints.
Applying the $\ReLU$ activation to an affine function introduces at most one additional
breakpoint, namely the root of the affine expression.
Since the network is one-dimensional, all roots are rational and can be computed exactly,
and the total number of breakpoints remains polynomial in the network size.
\end{nestedproof}

A crucial property of one-dimensional continuous functions is that the image of an interval
is again an interval.

\begin{lemma}
Let \(f:[0,1]\to[0,1]\) be continuous and affine on an interval \([a,b]\).
Then
\(
f([a,b]) = [\min(f(a),f(b)),\max(f(a),f(b))].
\)
\end{lemma}

As a consequence, for any interval \(I\subseteq[0,1]\), the image \(f(I)\) can be computed
by splitting \(I\) along the breakpoints of \(f\), applying the above formula on each
subinterval, and taking the union of the resulting intervals.

We now describe a polynomial-time algorithm for bounded-time reachability.

We maintain, for each time step \(t\), a set \(S_t\subseteq[0,1]\) representing all values
reachable at time \(t\).
Each \(S_t\) is stored as a finite union of disjoint closed intervals.

\medskip
\paragraph*{Algorithm.}
\begin{enumerate}
\item Compute the piecewise-affine representation of \(f=\NN\) with breakpoints
\(p_0,\dots,p_m\).
\item Initialize \(S_0 := R_0\).
\item For \(t=0,1,\dots,T\):
  \begin{enumerate}
  \item If \(S_t\cap R\neq\emptyset\), return \textsc{yes}.
  \item Compute \(S_{t+1} := f(S_t)\) as follows:
    \begin{itemize}
    \item For each interval \(I=[u,v]\) in \(S_t\), split \(I\) at the breakpoints
    \(p_i\), obtaining subintervals contained in a single affine piece.
    \item On each such subinterval \([\alpha,\beta]\), compute the image interval
    \([\min(f(\alpha),f(\beta)),\max(f(\alpha),f(\beta))]\).
    \item Take the union of all image intervals and merge overlaps to obtain \(S_{t+1}\).
    \end{itemize}
  \end{enumerate}
\item If no intersection with \(R\) is found up to time \(T\), return \textsc{no}.
\end{enumerate}


\begin{lemma}
For every \(t\le T\), the set \(S_t\) computed by the algorithm is exactly the set of points
reachable at time \(t\) from \(R_0\) under iteration of \(f\).
\end{lemma}

\begin{nestedproof}
By induction on \(t\).
For \(t=0\), \(S_0=R_0\) by definition.
Assume the claim holds for \(t\).
Then
\(
S_{t+1} = f(S_t) = \{f(x)\mid x\in S_t\},
\)
which is exactly the set of states reachable at time \(t+1\).
The construction of \(S_{t+1}\) is exact because \(f\) is affine on each piece and continuous
on \([0,1]\), so all extrema on subintervals are attained at endpoints.
\end{nestedproof}

Thus the algorithm answers \textsc{yes} if and only if there exists
\(x\in R_0\) and \(t\le T\) such that \(f^{[t]}(x)\in R\).


Let \(m\) be the number of affine pieces of \(f\).
Each set \(S_t\) is a union of at most \(O(m)\) intervals.
Each iteration step performs \(O(m)\) interval splits and merges, which can be carried out
in \(O(m\log m)\) time.

Since \(T\) is given in unary, the total running time is polynomial in the size of the input.
We obtain the following result.

This result critically relies on the one-width/one-dimensional setting.
In higher dimensions, images of convex sets under $\ReLU$ networks need not be convex or even
connected, and the above interval-propagation approach no longer applies.

\end{proof}

\section{Summary and open problems}
\label{sec:conclusion}

We have pinned down the complexity of verification for deep narrow
$\ReLU$ networks. On the discrete encoding set, width $3$ suffices
to encode $\textsc{3-Sat}$, yielding $\NP$-completeness
(Theorem~\ref{thm:np-complete-encoded-zero}), and the same
construction delivers $\coNP$-completeness
(Theorem~\ref{thm:np-conp}), $\#\P$-completeness
(Theorem~\ref{thm:sharp-p}), $\PP$-completeness
(Theorem~\ref{thm:pp}), and \textsc{Max-3Sat}-gap inapproximability
(Theorem~\ref{thm:inapprox}) as corollaries. At width~$4$, the
fractal preprocessing gadget $F_\varepsilon^{[n]}$ lifts the
single-quantifier results to arbitrary inputs in~$[0,1]$
(Theorem~\ref{thm:lifted} and Corollaries~\ref{cor:np-conp-lifted},
\ref{cor:uniform-lifted}, \ref{cor:inapprox-lifted}). Combined with
the polynomial-time decidability at width~$1$, both in the
single-evaluation model (Proposition~\ref{prop:width-one}) and in the
iterated one (Theorem~\ref{thdimone}), this gives a tight complexity
classification on the encoding set, where only width~$2$ remains
open; for arbitrary inputs in~$[0,1]$, widths~$2$ and~$3$ remain open.

The heart of our construction is the map
$F_\varepsilon\colon[0,1]\to[0,1]$: a continuous, width-$2$
$\ReLU$-realisable self-map whose $n$-fold iterate
$F_\varepsilon^{[n]}$ vanishes precisely on an
$\varepsilon$-neighbourhood of the $2^n$-point discrete encoding
set $\onethreefouradicn{n}$. 
The construction is genuinely
\emph{fractal}: the encoding set is the $n$-th level of a
Cantor-like subset of $[0,1]$, and $F_\varepsilon$ acts as the
branching rule of that fractal, refining the level at each
iteration. This lets a continuous $\ReLU$ gadget encode a discrete
exponential-size search space while consuming only two units of
width. We expect the technique to transfer to other
width-constrained models of continuous computation: any setting
where one must reduce a discrete combinatorial problem to a
continuous verification problem without inflating the state
dimension.

Several concrete open problems
remain: (i)~the complexity of width-$2$ verification (conjectured
$\NP$-hard); (ii)~whether the $\#\P$- and $\PP$-completeness
results extend from the encoding set to $[0,1]$, which would
require an injective filter; and (iii)~whether the lifting to
$[0,1]$ can be carried out at width~$3$.

\section*{Acknowledgements}
O. Bournez and J. Cohen received  support from PGMO program, and O. Bournez from ANR Project OTC ANR-24-CE48-0335. A. Wurm received partial support from the BTU Graduate Research School. 


  \end{document}

%% file: fig-lead.tex
\begin{figure}
\begin{tabular}{@{}c c@{}}
\adjustbox{valign=c}{%
\begin{tikzpicture}  

\node (x_in) at (0,0.5) {$x$};

\draw[->] (x_in) to node[pos=0.8, yshift=4pt] {\tiny 1} (1,0.5);
\draw[->] (x_in) to node[pos=0.8, yshift=4pt] {\tiny $8$} (1,-0.5);
\myrect{1}{0.8}{1.5}{-0.8}{$\ReLU$}
\myrect{1}{-0.2}{1.5}{-0.8}{$\ReLU_{-6}$}

\draw[->] (2.5,0.5) to node[pos=0.8, yshift=4pt] {\tiny 1} (3.5,0.5);
\draw[->] (2.5,0.5) to node[pos=0.8, yshift=4pt] {\tiny 8} (3.5,-0.5);
\draw[->] (2.5,-0.6) to node[pos=0.8, yshift=-1pt,below] {\tiny -1} (3.5,-0.6);


\myrect{3.5}{0.8}{1.5}{-0.8}{$\ReLU$}
\myrect{3.5}{-0.2}{1.5}{-0.8}{$\ReLU_{-4}$}

 \draw[->] (5,0.5) to node[pos=0.8, yshift=4pt] {\tiny 1} (6,0.5);
\draw[->] (5,0.5) to node[pos=0.8, yshift=4pt] {\tiny 1} (6,-0.5);
\draw[->] (5,-0.6) to node[pos=0.8, yshift=-1pt,below] {\tiny 1} (6,-0.6);

 
 \myrect{6}{0.8}{1.5}{-0.8}{$\ReLU$}
\myrect{6}{-0.2}{1.5}{-0.8}{$\ReLU_{1}$}

 
  \draw[->] (7.5,0.5) to node[pos=0.8, yshift=4pt] {\tiny $x$} (8,0.5);
  \draw[->] (7.5,-0.6) to node[pos=1.2, yshift=4pt] {\tiny $\lead(x)$} (8,-0.6);
 
 
 \end{tikzpicture} %
 }
&
\adjustbox{valign=c}{%
\begin{tikzpicture}[scale=0.7]
  \draw[->] (-0.02,0) -- (4.08,0) node[right] {$x$};
  \draw[->] (0,-0.02) -- (0,3.25) node[above] {$\lead(x)$};

  \foreach \x/\label in {0/0,1/$\frac14$,2/$\frac12$,3/$\frac34$,4/1}
    \draw (\x,0.05) -- (\x,-0.05) node[below=2pt] {\label};

  \foreach \y/\label in {1/1,3/3}
    \draw (0.02,\y) -- (-0.02,\y) node[left=2pt] {\label};

  \draw[dashed,gray] (0,1) -- (2,1);
  \draw[dashed,gray] (0,3) -- (4,3);

  \draw[thick]   (0,1) -- (2,1) ;
   \draw[thick]      (2,1)  --  (3,3) -- (4,3);
   \draw[thick]     (3,3) -- (4,3);


  \fill (0.25,1) circle (0.3pt);
  \fill (0.75,1) circle (0.3pt);
  \fill (0.75,3) circle (0.3pt);
  \fill (1,3) circle (0.3pt);
\end{tikzpicture}
}
\end{tabular}
\caption{Network $\NNlead$.  Boxes represent $\ReLU$ neurons; subscripts indicate biases, edge
labels indicate weights, and missing labels correspond to bias \(0\) and weight~\(1\).}
\label{fig:lead}
\end{figure}

%% file: fig-shift.tex
\begin{figure}
\centering
\begin{tabular}{@{}c c@{}}
\adjustbox{valign=c}{%
\begin{tikzpicture}     

\node (x_in) at (0,0.5) {$x$};
\draw[->] (x_in) to node[midway, above] {\tiny 1} (1,0.5);
\draw[->] (x_in) to node[midway, above] {\tiny $4$} (1,-0.5);
\myrect{1}{0.8}{1.5}{-0.8}{$\ReLU$}
\myrect{1}{-0.2}{1.5}{-0.8}{$\ReLU_{-1}$}

\draw[->] (2.5,0.5) to node[midway, above] {\tiny 1} (3.5,0.5);
\draw[->] (2.5,0.5) to node[midway, above] {\tiny -4} (3.5,-0.5);
\draw[->] (2.5,-0.6) to node[midway, below] {\tiny -1} (3.5,-0.6);


\myrect{3.5}{0.8}{1.5}{-0.8}{$\ReLU$}
\myrect{3.5}{-0.2}{1.5}{-0.8}{$\ReLU_{3}$}

 \draw[->] (5,0.5) to node[midway, above] {\tiny 1} (6,0.5);
\draw[->] (5,0.5) to node[midway, above] {\tiny -8} (6,-0.5);
\draw[->] (5,-0.6) to node[midway, below] {\tiny -1} (6,-0.6);

\myrect{6}{0.8}{1.5}{-0.8}{$\ReLU$}
\myrect{6}{-0.2}{1.5}{-0.8}{$\ReLU_{6}$}

\draw[->] (7.5,0.5) to node[midway, above] {\tiny 1} (8.5,0.5);
\draw[->] (7.5,0.5) to node[midway, above] {\tiny 4} (8.5,-0.5);
\draw[->] (7.5,-0.6) to node[midway, below] {\tiny 1} (8.5,-0.6);

 \myrect{8.5}{0.8}{1.5}{-0.8}{$\ReLU$}
\myrect{8.5}{-0.2}{1.5}{-0.8}{$\ReLU_{-3}$}

 
\draw[->] (10,0.5) to node[pos=1, right] {\tiny $x$} (10.5,0.5);
\draw[->] (10,-0.6) to node[pos=1, right] {\tiny $\shift(x)$} (10.5,-0.6);
 
 
 \end{tikzpicture}  %
}
&
\adjustbox{valign=c}{%
\begin{tikzpicture}[scale=2]
  \draw[->] (-0.02,0) -- (1.08,0) node[right] {$x$};
  \draw[->] (0,-0.02) -- (0,1.08) node[above] {$f(x)$};

  \foreach \x/\label in {0/0,0.25/$\frac14$,0.5/$\frac12$,0.75/$\frac34$,1/1}
    \draw (\x,0.02) -- (\x,-0.02) node[below=2pt] {\label};

  \foreach \y/\label in {0/0,1/1}
    \draw (0.02,\y) -- (-0.02,\y) node[left=2pt] {\label};

  \draw[thick] (0,0) -- (0.25,0);
    \draw[thick] (0.25,0) -- (0.5,1);

  \draw[thick] (0.5,1) -- (0.75,0);
  \draw[thick] (0.75,0) -- (1,1);

  \end{tikzpicture}}
\end{tabular}

\caption{Network $\NNshift$ computing $x\mapsto(x,\shift(x))$. The top
row of neurons forwards $x$ unchanged through every layer, while the
bottom row accumulates the $\ReLU$ hinges that assemble the
piecewise-affine map $\shift$. Boxes are $\ReLU$ neurons; a subscript
denotes the neuron's bias, an edge label its weight, and an absent
label means bias~$0$ and weight~$1$. Right: the graph of the resulting
map on $[0,1]$.}
\label{fig:shift}
\end{figure}

%% file: fig-test.tex
\begin{figure}  
\centering
\begin{tabular}{@{}c c@{}}
\adjustbox{valign=c}{%
\begin{tikzpicture}[scale=1.5]
  \draw[->] (-0.03,0) -- (1.08,0) node[right] {$p$};
  \draw[->] (0,-0.03) -- (0,1.08) node[above] {$g(f(0,p))$};

  \foreach \x/\lab in {0/0,0.25/$\frac14$,0.5/$\frac12$,0.75/$\frac34$,1/1}
    \draw (\x,0.015) -- (\x,-0.015) node[below=2pt] {\lab};

  \foreach \y/\lab in {0/0,0.25/$\frac14$,0.5/$\frac12$,0.75/$\frac34$,1/1}
    \draw (0.015,\y) -- (-0.015,\y) node[left=2pt] {\lab};

  \draw[thick] (0,0) -- (1,1);

  \fill (0,0) circle (0.4pt);
  \fill (1,1) circle (0.4pt);
\end{tikzpicture}  %
}
&
\adjustbox{valign=c}{%
\begin{tikzpicture}[xscale=3,yscale=5]
  \draw[->] (-0.03,0) -- (1.08,0) node[right] {$p$};
  \draw[->] (0,-0.03) -- (0,0.28) node[above] {$g(f(1,p))$};

  \foreach \x/\lab in {0/0,0.25/$\frac14$,0.5/$\frac12$,0.75/$\frac34$,1/1}
    \draw (\x,0.01) -- (\x,-0.01) node[below=2pt] {\lab};

  \foreach \y/\lab in {0/0,0.125/$\frac18$,0.25/$\frac14$}
    \draw (0.01,\y) -- (-0.01,\y) node[left=2pt] {\lab};

  \draw[dashed,gray] (0.25,0) -- (0.25,0.25);
  \draw[dashed,gray] (0.75,0) -- (0.75,0.25);
  \draw[dashed,gray] (0,0.25) -- (0.25,0.25);
  \draw[dashed,gray] (0,0.125) -- (0.5,0.125);

  \draw[thick] (0,0) -- (0.25,0.25);
  \draw[thick] (0.25,0.25) -- (0.75,0);
  \draw[thick] (0.75,0) -- (1,0.25);

  \fill (0,0) circle (0.4pt);
  \fill (0.25,0.25) circle (0.4pt);
  \fill (0.75,0) circle (0.4pt);
  \fill (1,0.25) circle (0.4pt);
\end{tikzpicture}}
\end{tabular}

\caption{Network $\NNtest$. Illustration of the graphs of the function $g$ and of the composition $g\circ f$ for Lemma~\ref{lem:test}.}
\label{fig:test}
\end{figure}

%% file: fig-NP-3.tex
\begin{figure}
\begin{tikzpicture}     
 \node (x_in) at (-0.1,1.5) {$x_{j-1}$};   
\node (p_in) at (-0.1,-0.6) {$p_{\ell}$};
 
\draw[->] (x_in) to node[pos=0.8, yshift=4pt] {\tiny 4} (0.5,1.5);
\draw[->](p_in) --  (0.6,-0.6);

%
\myrectrotate{0.6}{1.9}{1}{-1.8}{$\NNlead$}
        \myrect{0.6}{-0.3}{1}{-0.7}{$\ReLU$}

  \draw[->] (1.6,1.5)  --   (2.5,1.5);   
  \draw[->](1.6,0.5) to node[pos=0.2, yshift=5pt] {\tiny $a_{j}$}  (2.5,0.5);   

\draw[->](1.6,0.5) to node[pos=0.8, yshift=5pt] {\tiny $-\frac12$}  (2.5,0.5);   

\draw[->](1.6,-0.6) --  (2.5,-0.6);   


\myrect{2.5}{1.9}{1.2}{-0.7}{$\ReLU$}
\myrect{2.5}{0.8}{1.3}{-0.7}{$\ReLU_{\frac32}$}
\myrect{2.5}{-0.3}{1.3}{-0.7}{$\ReLU$}

 \draw[->](3.8,1.5) --  (4.5,1.5);   
 \draw[->](3.8,0.5) to node[pos=0.2, yshift=5pt] {\tiny $b_{j}$}  (4.5,0.5);   
\draw[->](3.8,-0.6) --  (4.5,-0.6);

 \myrect{4.5}{1.9}{1}{-0.7}{$\ReLU$}
\myrect{4.5}{0.8}{1}{-0.7}{$\ellit_{\ell,j}$}
\myrect{4.5}{-0.3}{1}{-0.7}{$\ReLU$}

 \draw[->](5.5,1.5) --  (6.3,1.5);   
 \draw[->](5.5,0.5) to node[pos=0.2, yshift=5pt] {\tiny $e_{j}$}   (6.3,0.5);   
 \draw[->](5.5,-0.6) to node[pos=0.2, yshift=5pt] {\tiny $p_{\ell}$}   (6.3,-0.6);

 \myrect{6.3}{1.9}{1}{-0.7}{$\ReLU$}
\myrectrotate{6.3}{0.8}{1}{-1.8}{$\NNtest$}

 \draw[->](7.3,1.5) --  (8,1.5);   
 \draw[->](7.3,0.5) to node[pos=0.2, yshift=5pt] {\tiny $b_{j}$}  (8,0.5);   
\draw[->](7.3,-0.6) --  (8,-0.6);

   \myrectrotate{8}{1.9}{1}{-1.8}{$\NNlead$}
   \myrect{8}{-0.3}{1}{-0.7}{$\ReLU$}
 
  \draw[->] (9,1.5)  to node[pos=0.2, yshift=5pt] {\tiny $x_{j}$}    (9.9,1.5);   
  \draw[->](9,0.5) to node[pos=0.2, yshift=5pt] {\tiny $a_{j}$}  (9.9,0.5);    
  \draw[->](9,-0.6) --  (9.9,-0.6);     

  
   \myrectrotate{9.9}{1.9}{1}{-1.8}{$\NNshift'$}
   \myrect{9.9}{-0.3}{1}{-0.7}{$\ReLU$}

  \draw[->] (10.9,1.5)  to node[pos=0.13, yshift=5pt] {\tiny $x_{j}$}    (11.9,1.5);   
  \draw[->](10.9,0.5) to node[pos=0.13, yshift=5pt] {\tiny $a_{j}$}  (11.9,1.2);    
   \draw[->](10.9,0.5) to node[pos=0.8, yshift=5pt] {\tiny $4^{-n}$}  (11.9,1.2);    

  \draw[->](10.9,-0.6) --  (11.9,-0.6);     

  \myrect{11.9}{1.9}{1}{-0.7}{$\ReLU$}
\myrect{11.9}{0.8}{1}{-0.7}{$\ReLU$}
\myrect{11.9}{-0.3}{1}{-0.7}{$\ReLU$}

  \draw[->] (12.9,1.5)  to node[pos=0.8, yshift=5pt] {\tiny $x_{j}$}    (13.5,1.5);   
   \draw[->](12.9,-0.6)  to node[pos=0.8, yshift=5pt] {\tiny $p_{\ell}$}   (13.5,-0.6);

 
 \end{tikzpicture}  
\caption{Network $\NN_{\varphi,\ell,j}$ updating $(x_{j-1},p_{\ell})$ to
$(x_j,p_\ell)$ using ReLU neurons and the three components  $\NNlead$, $\NNtest$, and $\NNshift'$. 
For layout reasons, we identify $\NNshift'(x,\lead(x))$ with $\shift'(x)$,
corresponding to the equation $4x-\lead(x)=\shift'(x)$.
The values $(x_{j-1},p_{\ell})$ represent the variables $(x,p_{\ell})$ at the beginning
of the $j$-th iteration of the loop in $\NN_{\varphi,\ell}$.
Outputs are shown at the sources of the arrows, and weights at their endpoints.
The width of this sub-network is $3$.}

\label{fig:NP:network-4}
\end{figure}

%% file: fig-feps.tex
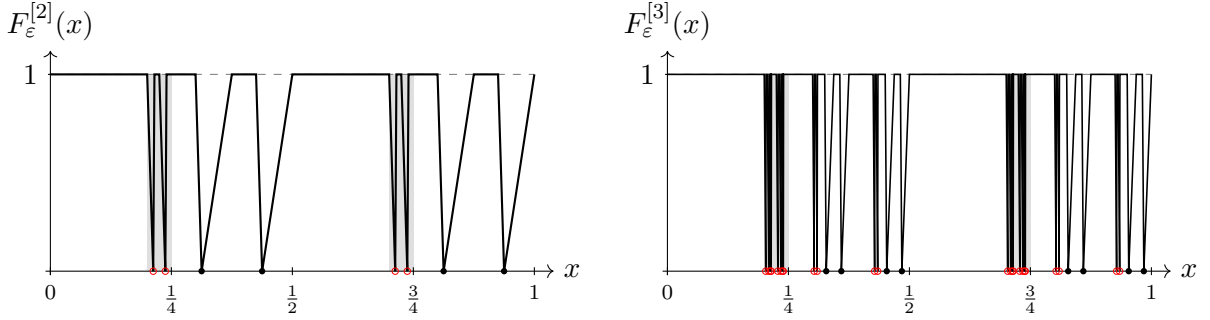
\begin{figure}[t]
\centering
\begin{minipage}{0.49\textwidth}
\centering
\begin{tikzpicture}[xscale=6.4,yscale=2.6,
    gz/.style={circle,fill=black,inner sep=0pt,minimum size=2.4pt},
    sz/.style={circle,draw=red,inner sep=0pt,minimum size=2.4pt,line width=0.35pt}]
  \fill[gray!25] ({1/4-1/20},0) rectangle (1/4,1);
  \fill[gray!25] ({3/4-1/20},0) rectangle (3/4,1);
  \draw[->] (-0.01,0) -- (1.04,0) node[right] {$x$};
  \draw[->] (0,-0.02) -- (0,1.12) node[above] {$F_\varepsilon^{[2]}(x)$};
  \draw[dashed,gray] (0,1) -- (1,1);
  \draw[thick] (0.00000,1.0000) -- (0.20000,1.0000) -- (0.21250,0.0000) -- (0.21500,1.0000) -- (0.22500,1.0000) -- (0.23750,0.0000) -- (0.24000,1.0000) -- (0.30000,1.0000) -- (0.31250,0.0000) -- (0.37500,1.0000) -- (0.42500,1.0000) -- (0.43750,0.0000) -- (0.50000,1.0000) -- (0.70000,1.0000) -- (0.71250,0.0000) -- (0.71500,1.0000) -- (0.72500,1.0000) -- (0.73750,0.0000) -- (0.74000,1.0000) -- (0.80000,1.0000) -- (0.81250,0.0000) -- (0.87500,1.0000) -- (0.92500,1.0000) -- (0.93750,0.0000) -- (1.00000,1.0000);
  \node[gz] at (0.31250,0) {};
  \node[gz] at (0.43750,0) {};
  \node[gz] at (0.81250,0) {};
  \node[gz] at (0.93750,0) {};
  \node[sz] at (0.21250,0) {};
  \node[sz] at (0.23750,0) {};
  \node[sz] at (0.71250,0) {};
  \node[sz] at (0.73750,0) {};
  \draw (0.005,1) -- (-0.005,1) node[left] {$1$};
  \foreach \x/\lab in {0/0,0.25/{\frac14},0.5/{\frac12},0.75/{\frac34},1/1}{
    \draw (\x,0.02) -- (\x,-0.02) node[below,font=\scriptsize] {$\lab$};
  }
\end{tikzpicture}
\end{minipage}\hfill
\begin{minipage}{0.49\textwidth}
\centering
\begin{tikzpicture}[xscale=6.4,yscale=2.6,
    gz/.style={circle,fill=black,inner sep=0pt,minimum size=2.4pt},
    sz/.style={circle,draw=red,inner sep=0pt,minimum size=2.4pt,line width=0.35pt}]
  \fill[gray!25] ({1/4-1/20},0) rectangle (1/4,1);
  \fill[gray!25] ({3/4-1/20},0) rectangle (3/4,1);
  \draw[->] (-0.01,0) -- (1.04,0) node[right] {$x$};
  \draw[->] (0,-0.02) -- (0,1.12) node[above] {$F_\varepsilon^{[3]}(x)$};
  \draw[dashed,gray] (0,1) -- (1,1);
  \draw[semithick] (0.00000,1.0000) -- (0.20000,1.0000) -- (0.20312,0.0000) -- (0.20375,1.0000) -- (0.20625,1.0000) -- (0.20938,0.0000) -- (0.21000,1.0000) -- (0.21300,1.0000) -- (0.21313,0.0000) -- (0.21375,1.0000) -- (0.21425,1.0000) -- (0.21438,0.0000) -- (0.21500,1.0000) -- (0.22500,1.0000) -- (0.22812,0.0000) -- (0.22875,1.0000) -- (0.23125,1.0000) -- (0.23438,0.0000) -- (0.23500,1.0000) -- (0.23800,1.0000) -- (0.23813,0.0000) -- (0.23875,1.0000) -- (0.23925,1.0000) -- (0.23938,0.0000) -- (0.24000,1.0000) -- (0.30000,1.0000) -- (0.30312,0.0000) -- (0.30375,1.0000) -- (0.30625,1.0000) -- (0.30938,0.0000) -- (0.31000,1.0000) -- (0.32500,1.0000) -- (0.32812,0.0000) -- (0.34375,1.0000) -- (0.35625,1.0000) -- (0.35938,0.0000) -- (0.37500,1.0000) -- (0.42500,1.0000) -- (0.42812,0.0000) -- (0.42875,1.0000) -- (0.43125,1.0000) -- (0.43438,0.0000) -- (0.43500,1.0000) -- (0.45000,1.0000) -- (0.45312,0.0000) -- (0.46875,1.0000) -- (0.48125,1.0000) -- (0.48438,0.0000) -- (0.50000,1.0000) -- (0.70000,1.0000) -- (0.70312,0.0000) -- (0.70375,1.0000) -- (0.70625,1.0000) -- (0.70937,0.0000) -- (0.71000,1.0000) -- (0.71300,1.0000) -- (0.71313,0.0000) -- (0.71375,1.0000) -- (0.71425,1.0000) -- (0.71437,0.0000) -- (0.71500,1.0000) -- (0.72500,1.0000) -- (0.72813,0.0000) -- (0.72875,1.0000) -- (0.73125,1.0000) -- (0.73438,0.0000) -- (0.73500,1.0000) -- (0.73800,1.0000) -- (0.73813,0.0000) -- (0.73875,1.0000) -- (0.73925,1.0000) -- (0.73938,0.0000) -- (0.74000,1.0000) -- (0.80000,1.0000) -- (0.80312,0.0000) -- (0.80375,1.0000) -- (0.80625,1.0000) -- (0.80937,0.0000) -- (0.81000,1.0000) -- (0.82500,1.0000) -- (0.82812,0.0000) -- (0.84375,1.0000) -- (0.85625,1.0000) -- (0.85938,0.0000) -- (0.87500,1.0000) -- (0.92500,1.0000) -- (0.92812,0.0000) -- (0.92875,1.0000) -- (0.93125,1.0000) -- (0.93437,0.0000) -- (0.93500,1.0000) -- (0.95000,1.0000) -- (0.95312,0.0000) -- (0.96875,1.0000) -- (0.98125,1.0000) -- (0.98438,0.0000) -- (1.00000,1.0000);
  \node[gz] at (0.32812,0) {};
  \node[gz] at (0.35938,0) {};
  \node[gz] at (0.45312,0) {};
  \node[gz] at (0.48438,0) {};
  \node[gz] at (0.82812,0) {};
  \node[gz] at (0.85938,0) {};
  \node[gz] at (0.95312,0) {};
  \node[gz] at (0.98438,0) {};
  \node[sz] at (0.20312,0) {};
  \node[sz] at (0.20938,0) {};
  \node[sz] at (0.21313,0) {};
  \node[sz] at (0.21438,0) {};
  \node[sz] at (0.22812,0) {};
  \node[sz] at (0.23438,0) {};
  \node[sz] at (0.23813,0) {};
  \node[sz] at (0.23938,0) {};
  \node[sz] at (0.30312,0) {};
  \node[sz] at (0.30938,0) {};
  \node[sz] at (0.42812,0) {};
  \node[sz] at (0.43438,0) {};
  \node[sz] at (0.70312,0) {};
  \node[sz] at (0.70937,0) {};
  \node[sz] at (0.71313,0) {};
  \node[sz] at (0.71437,0) {};
  \node[sz] at (0.72813,0) {};
  \node[sz] at (0.73438,0) {};
  \node[sz] at (0.73813,0) {};
  \node[sz] at (0.73938,0) {};
  \node[sz] at (0.80312,0) {};
  \node[sz] at (0.80937,0) {};
  \node[sz] at (0.92812,0) {};
  \node[sz] at (0.93437,0) {};
  \draw (0.005,1) -- (-0.005,1) node[left] {$1$};
  \foreach \x/\lab in {0/0,0.25/{\frac14},0.5/{\frac12},0.75/{\frac34},1/1}{
    \draw (\x,0.02) -- (\x,-0.02) node[below,font=\scriptsize] {$\lab$};
  }
\end{tikzpicture}
\end{minipage}
\caption{The iterates $F_\varepsilon^{[2]}$ (left) and $F_\varepsilon^{[3]}$ (right)
of the continuous map $F_\varepsilon$ of Equation~\eqref{eq:Feps}, drawn for
$\varepsilon=\tfrac1{20}$; the repair intervals $J_\varepsilon$ are shaded.
Black dots are the genuine zeros, namely the elements of $\onethreefouradicn{n}$
(Lemma~\ref{lem:coincidence}). Red circles are the spurious zeros forced by
continuity; each lies close to a shorter encoding (Lemma~\ref{lem:localisation}).}
\label{fig:Feps-iterates}
\end{figure}